\documentclass[final,1p,times]{elsarticle}
\usepackage{amssymb,amsmath}
\usepackage{qcircuit}
\usepackage{graphicx}
\usepackage{pstricks,pst-node,pst-tree}
\usepackage{color}
\usepackage{macros}
\usepackage{stmaryrd}
\usepackage{url}
\usepackage{proof}

\definecolor{dark red}{rgb}{0.545098,0.000000,0.000000}
\definecolor{dark green}{rgb}{0.000000,0.392157,0.000000}
\definecolor{dark blue}{rgb}{0.000000,0.000000,0.545098}
\definecolor{brown}{rgb}{0.647059,0.164706,0.164706}

\def \leaveout#1{}

\def \Math#1{{ $#1$}}

\def\sset#1{\{#1\}}

\def\Nat{{\mathbb N}}
\def\Exp{{\mathbb E}}
\def\comp{\rhd\!\!\!\lhd}
\newcommand{\Char}{{\bf 1}}
\newcommand{\triple}[3]{\{#1\}\;#2\;\{#3\}}

\newcommand{\Assn}{{\textbf{Assn}}}
\newcommand{\Denote}[1]{[\![#1 ]\!]}

\def\PState{{\textbf{POVD}}}
\newcommand{\uclosed}{\sf uclosed}
\def\state{\sigma}
\def\qstate{\rho}
\def\povd{\mu}

\def\Qimp{{\textbf{QIMP}}}

\def\Tru{{\textbf{T}}}
\def\true{{\textbf{true}}}
\def\false{{\textbf{false}}}

\def\Cvar{{\textbf{Cvar}}}
\def\Qvar{{\textbf{Qvar}}}
\def\Aexp{{\textbf{Aexp}}}
\def\Bexp{{\textbf{Bexp}}}
\def\Com{{\textbf{Com}}}
\def\Skip{{\textbf{skip}}}
\def\nil{\textbf{nil}}
\def\abort{\textbf{abort}}
\def\If{\textbf{if}~}
\def\Then{~\textbf{then}~}
\def\Else{~\textbf{else}~}
\def\While{\textbf{while}~}
\def\Do{~\textbf{do}~}

\def\Assn{{\textbf{Assn}}}

\def\pair#1{\langle#1\rangle}

\def\CH{{\cal H}}
\def\h{\ensuremath{\mathcal{H}}}
\def\lh{\ensuremath{\mathcal{L(H)}}}

\def\CNOT{{\textit{CNOT}}}

\newcommand{\cH}{{\cal H}}
\renewcommand{\>}{\ensuremath{\rangle}}
\newcommand{\<}{\ensuremath{\langle}}
\newcommand{\tU}{{ U}}
\newcommand{\tM}{{ M}}
\newcommand{\tr}{{\rm tr}}
\newcommand{\Lab}{{l}}
\newcommand{\Id}{{\it Id}}
\newcommand{\cond}[1]{{\color{blue}#1}}

\def \VS{\vspace{2mm}}

\begin{document}

\begin{frontmatter}
  \title{Formal Semantics of a Classical-Quantum Language}
  
\author[ecnu]{Yuxin Deng\corref{deng}}
\ead{yxdeng@sei.ecnu.edu.cn}

\author[uts]{Yuan Feng}
\ead{Yuan.Feng@uts.edu.au}


\address[ecnu]{Shanghai Key Laboratory of Trustworthy Computing,
East China Normal University}

\address[uts]{University of Technology Sydney, Australia}

\begin{abstract}
We investigate the formal semantics of a simple imperative  language that has both classical and quantum constructs.  More specifically, we provide an operational semantics, a denotational semantics and two Hoare-style proof systems: an abstract one and a concrete one. The two proof systems are satisfaction-based, as inspired by the program logics of Barthe et al for probabilistic programs. The abstract proof system turns out to be sound and relatively complete, while the concrete one is sound only.
\end{abstract}


\end{frontmatter}

\section{Introduction}

Programming is the core of software development, but it is also an inherently error-prone activity. The likelihood of errors will even be significantly higher when programming with a quantum computer, as the techniques used for classical programming are, unfortunately, hard to apply to quantum computers because quantum systems are essentially different from classical ones. Thus, there is a pressing need to provide verification and analysis techniques for reasoning about the correctness of quantum programs. Furthermore, these techniques would also be very useful for compiling and optimising quantum programs.
	
Among other techniques, Hoare logic~\cite{hoare1969axiomatic} provides a syntax-oriented proof system to reason about program correctness. After decades of development, Hoare logic has been successfully applied in analysis of programs with non-determinism, recursion, parallel execution, etc~\cite{apt2019fifty,apt2010verification}. It was also extended to programming languages with probabilistic features. Remarkably, as the program states for probabilistic languages are (sub)distributions over evaluations of program variables, the extension naturally follows two different approaches, depending on how assertions of probabilistic states are defined. The first one
takes subsets of distributions as (qualitative) assertions, similar to the non-probabilistic case, and the \emph{satisfaction} relation between distributions and assertions is then just the ordinary membership~\cite{ramshaw1979formalizing,den2002verifying,chadha2007reasoning,BEGGHS18}. In contrast, the other approach takes
non-negative functions on evaluations as (quantitative) assertions. Consequently, one is concerned with the \emph{expectation} of a distribution satisfying an assertion~\cite{morgan1996probabilistic,mciver2005abstraction,olmedo2016reasoning,kozen1981semantics,kozen1985probabilistic}.

In the past two decades, several Hoare-type logic systems for quantum programs (QHL) have been proposed, also following the two approaches as in the probabilistic setting. 

\textbf{Expectation-based QHLs}. In the logic systems proposed in~\cite{Yin12,ying2016foundations,ying2018reasoning,li2019quantum} for purely quantum programs, the assertions $P$ and $Q$ in a Hoare triple $\{P\} S \{Q\}$ are both positive operators with the eigenvalues lying in $[0,1]$, and such a triple is valid in the sense of total correctness if for any initial quantum state $\rho$, $\tr(P\rho) \leq \tr(Q\rho')$ where $\rho'$ is the final state obtained by executing $S$ on $\rho$, and $\tr(P\rho)$ denotes the expectation/degree of $\rho$ satisfying $P$ (or physically, the average outcome when measuring $\rho$ according to the projective measurement determined by $P$). This definition captures the idea that the precondition $P$ (on the initial state) provides a lower bound on the degree of satisfaction of the postcondition $Q$ (on the finial state). 
More recently, this type of expectation-based QHL has been extended to quantum programs with classical variables~\cite{feng2020quantum}
as well as distributed quantum programs with classical communication~\cite{feng2021verification}.
These logic systems have proven to be useful in describing and verifying correctness of a wide range of quantum algorithms such as Shor's algorithm~\cite{Sho94},  Grover's algorithm~\cite{Gro96}, etc. Moreover, they are theoretically elegant: all of these systems are (relatively) complete in the sense that every semantically valid Hoare triple can be deduced from the corresponding proof system.

\textbf{Limit of the expectation-based approach}. However, the expectation-based quantum Hoare logic systems proposed in the literature all suffer from the following \emph{expressiveness} problems.
\begin{itemize}
 \item	Unlike the classical boolean-valued assertions, the positive operator assertions cannot exclude undesirable quantum inputs. Note that the correctness of many algorithms such as quantum teleportation (an EPR pair is assumed as part of the input) and phase estimation (the corresponding eigenstate is given) all have some restrictions on the input. The expressiveness of expectation-based QHL might be limited, as the following artificial example suggests.
	Let 
	$$S\equiv  x := M_{\pm}[q]$$
	where $M_{\pm}$ is the measurement according to the $|\pm\>$ basis. 
	Obviously, starting with $|0\>$, the program ends at $|+\>$ with probability 0.5. However, expectation-based QHL cannot describe this property: the natural candidate
	\begin{equation}\label{eq:cexample}
	\{0.5 |0\>\<0|\}\ S\ \{|+\>\<+|\}
	\end{equation}
	is actually not valid. The reason is, it somehow \emph{over-specifies} the correctness: in addition to the above desirable property, it also puts certain requirement for \emph{all other} possible input states. To see this, take $\rho = |-\>\<-|$. Then $\Denote{S}(\rho) = |-\>\<-|$. Thus $$\tr(0.5 |0\>\<0|\rho) = 0.25 > \tr(|+\>\<+|\Denote{S}(\rho)) =0,$$
	making the Hoare triple in Eq.~\eqref{eq:cexample} invalid.
	 \item	As logic operations such as conjunction and disjunction are difficult, if at all possible, to define for positive operators, complicated properties can only be analysed separately, making the verification process cumbersome. This has been pointed out in~\cite{BEGGHS18} for expectation-based probabilistic Hoare logics. The same is obviously true for expectation-based QHLs as well.
\end{itemize}

\textbf{Satisfaction-based QHLs}. An Ensemble Exogenous Quantum Propositional Logic (EEQPL) was proposed in~\cite{chadha2006reasoning} for a simple quantum language with bounded iteration. The assertions in EEQPL can access amplitudes of quantum states, which makes it very strong in expressiveness but also hinders its use in applications such as debugging, as amplitudes are not physically accessible through measurements. The completeness of EEQPL is only proven in a special case where all real and complex values involved range over a finite set. In contrast, the QHL proposed in~\cite{Kakutani:2009} takes as the assertion language an extended first-order logic with the primitives of applying a matrix on a set of qubits and computing the probability that a classical predicate is satisfied by the outcome of a quantum measurement. The proof system is shown to be sound, but no completeness result was established. 

Another way of defining satisfaction-based QHLs proposed in~\cite{zhou2019applied,unruh2019quantum} regard subspaces of the Hilbert space as assertions, and a quantum state $\rho$ satisfies an assertion $P$ iff the support (the image space of linear operators) of $\rho$ is included in $P$.  The subspace assertion 
makes it easy to describe and determine properties of quantum programs, but the expressive power of the assertions is limited: they only assert if a given quantum state lies completely within a subspace. Consequently,  quantum algorithms which succeed with certain probability cannot be verified in their logic systems.

\textbf{Contribution of the current paper}. In this paper, motivated by~\cite{BEGGHS18}, we propose two Hoare-style proof systems: an abstract one and a concrete one. 
It is worth noting that the imperative language we consider here involves both classical and quantum constructs. 
Our work distinguishes itself from the works on QHLs mentioned above in the following aspects:
\begin{itemize}
	\item \emph{Assertion language}. The assertions used in our logic systems are boolean-typed, so that they can be easily combined using logic operations such as disjunction and conjunction. On the other hand, all information used in the assertions are physically accessible: they can be obtained through quantum measurement applying on the program states.
	  For example, consider the protocol of superdense coding (see Section~\ref{sec:example}). In order to send a message of two bits stored in two variables $x_0x_1$, Alice actually sends a qubit to Bob. From the received quantum information Bob can recover the message $y_0y_1$. The property we would expect is $x_0=y_0 \wedge x_1=y_1$. Indeed, by letting $SC$ be a quantum program to implement the protocol, we can prove that the following judgement is derivable:
          \[\triple{\true} {SC}{\Box(x_0=y_0 \wedge x_1=y_1)}\]
          where the precondition $\true$ is satisfied by any program state. Intuitively, the judgement says that the message received by Bob is always the same as that sent by Alice, no matter what  the initial program state is. There is no need to mention the concrete values of $x_0$ and $x_1$. This is a natural and concise way of specifying the correctness of the program $SC$. 
	\item \emph{Satisfaction-based complete QHL}. The existing satisfaction-based QHLs proposed in the literature all lack of completeness; the only exception is~\cite{zhou2019applied}, but as mentioned above, the assertions there are not expressive enough to verify probabilistic correctness.
          The abstract proof system proposed in the current work is shown to be sound and relatively complete, while the concrete proof system is sound only. By soundness, we mean that if the Hoare triple $\triple{P}{c}{Q}$ is derivable, then for any program state $\povd$ that satisfies $P$, the program state $\Denote{c}_\povd$ after the execution of command $c$ always satisfies $Q$. Completeness means the converse. We only have relative completeness because our assertion language allows for first-order logic operators such as implication. Note that in the presence of both classical and quantum variables, we represent each program state $\povd$ as a partial density operator valued distribution (POVD), and interpret a program as a transformer of POVDs. We establish a consistence result between the denotational semantics and the small-step operational semantics based on POVDs.
\end{itemize}

The rest of the paper is structured as follows. In Section~\ref{sec:pre} we recall some basic notations from linear algebra and quantum mechanics. In Section~\ref{sec:qimp} we define the syntax and operational semantics of a simple classical-quantum imperative  language. In Section~\ref{sec:abs} we present an abstract proof system and show its soundness and relative completeness. In Section~\ref{sec:conc} we provide a concrete proof system. In Section~\ref{sec:example} we use the example of superdense coding to illustrate the concrete proof system. Finally, we conclude in Section~\ref{sec:clu}.

\section{Preliminaries}\label{sec:pre}
We briefly recall some basic notations
from linear algebra and quantum mechanics which are needed in this paper. 
For more details, we refer to \cite{NC00}.

A {\it Hilbert space} $\h$ is a complete vector space with an inner
product $\langle\cdot|\cdot\rangle:\h\times \h\rightarrow \mathbf{C}$
such that 
\begin{enumerate}
\item
$\langle\psi|\psi\rangle\geq 0$ for any $|\psi\>\in\h$, with
equality if and only if $|\psi\rangle =0$;
\item
$\langle\phi|\psi\rangle=\langle\psi|\phi\rangle^{\ast}$;
\item
$\langle\phi|\sum_i c_i|\psi_i\rangle=
\sum_i c_i\langle\phi|\psi_i\rangle$,
\end{enumerate}
where $\mathbf{C}$ is the set of complex numbers, and for each
$c\in \mathbf{C}$, $c^{\ast}$ stands for the complex
conjugate of $c$. For any vector $|\psi\rangle\in\h$, its
length $|||\psi\rangle||$ is defined to be
$\sqrt{\langle\psi|\psi\rangle}$, and it is said to be {\it normalised} if
$|||\psi\rangle||=1$. Two vectors $|\psi\>$ and $|\phi\>$ are
{\it orthogonal} if $\<\psi|\phi\>=0$. An {\it orthonormal basis} of a Hilbert
space $\h$ is a basis $\{|i\rangle\}$ where each $|i\>$ is
normalised and any pair of them are orthogonal.

Let $\lh$ be the set of linear operators on $\h$.
For any $A\in \lh$, $A$ is {\it Hermitian} if $A^\dag=A$ where $A^\dag$ is the adjoint operator of $A$ such that
$\<\psi|A^\dag|\phi\>=\<\phi|A|\psi\>^*$ for any
$|\psi\>,|\phi\>\in\h$.
A linear operator $A\in \lh$ is {\it unitary} if $A^\dag A=A A^\dag=I_\h$ where $I_\h$ is the
identity operator on $\h$. 
The {\it  trace} of $A$ is defined as $\tr(A)=\sum_i \<i|A|i\>$ for some
given orthonormal basis $\{|i\>\}$ of $\h$.
A linear operator $A\in \lh$ is {\it positive} if $\<\phi|A|\phi\> \geq 0$ for any state $|\phi\> \in\h$.
The \emph{L\"{o}wner order} $\sqsubseteq$ on the set of Hermitian operators on $\h$ is defined by letting $A\sqsubseteq B$ iff $B-A$ is positive.

Let $\h_1$ and $\h_2$ be two Hilbert spaces. Their {\it tensor product} $\h_1\otimes \h_2$ is
defined as a vector space consisting of
linear combinations of the vectors
$|\psi_1\psi_2\rangle=|\psi_1\>|\psi_2\rangle =|\psi_1\>\otimes
|\psi_2\>$ with $|\psi_1\rangle\in \h_1$ and $|\psi_2\rangle\in
\h_2$. Here the tensor product of two vectors is defined by a new
vector such that
$$\left(\sum_i \lambda_i |\psi_i\>\right)\otimes
\left(\sum_j\mu_j|\phi_j\>\right)=\sum_{i,j} \lambda_i\mu_j
|\psi_i\>\otimes |\phi_j\>.$$ Then $\h_1\otimes \h_2$ is also a
Hilbert space where the inner product is defined as the following:
for any $|\psi_1\>,|\phi_1\>\in\h_1$ and $|\psi_2\>,|\phi_2\>\in
\h_2$,
$$\<\psi_1\otimes \psi_2|\phi_1\otimes\phi_2\>=\<\psi_1|\phi_1\>_{\h_1}\<
\psi_2|\phi_2\>_{\h_2}$$ where $\<\cdot|\cdot\>_{\h_i}$ is the inner
product of $\h_i$.

By applying quantum gates to qubits, we can change their
states. For example, the Hadamard gate (H gate) can be applied
on a single qubit, while the CNOT gate can be applied on two
qubits. Some commonly used gates and their
representation in terms of matrices are as follows.
$$\CNOT=\left(%
\begin{array}{cccc}
  1 & 0 & 0 & 0 \\
  0 & 1 & 0 & 0 \\
  0 & 0 & 0 & 1 \\
  0 & 0 & 1 & 0
\end{array}%
\right),$$
\[
H=\frac{1}{\sqrt{2}}\left(%
\begin{array}{cc}
  1 & 1 \\
  1 & -1 \\
\end{array}%
\right),\ \  I_2=\left(%
\begin{array}{cc}
  1 & 0 \\
  0 & 1 \\
\end{array}%
\right),\ \
X=\left(%
\begin{array}{cc}
  0 & 1 \\
  1 & 0 \\
\end{array}%
\right),\ \ Z=\left(%
\begin{array}{cc}
  1 & 0 \\
  0 & -1 \\
\end{array}%
\right).
\]

According to von Neumann's formalism of quantum mechanics
\cite{vN55}, an isolated physical system is associated with a
Hilbert space which is called the {\it state space} of the system. A {\it pure state} of a
quantum system is a normalised vector in its state space, and a
{\it mixed state} is represented by a density operator on the state
space. Here a \emph{density operator} $\rho$ on Hilbert space $\h$ is a
positive linear operator such that $\tr(\rho)= 1$. A \emph{partial density operator} $\rho$ is a positive linear operator with $\tr(\rho)\leq 1$.

The evolution of a closed quantum system is described by a unitary
operator on its state space: if the states of the system at times
$t_1$ and $t_2$ are $\rho_1$ and $\rho_2$, respectively, then
$\rho_2=U\rho_1U^{\dag}$ for some unitary operator $U$ which
depends only on $t_1$ and $t_2$. 

A quantum {\it measurement} is described by a
collection $\{M_m\}$ of measurement operators, where the indices
$m$ refer to the measurement outcomes. It is required that the
measurement operators satisfy the completeness equation
$\sum_{m}M_m^{\dag}M_m=I_\h$. If the system is in state $\rho$, then the probability
that measurement result $m$ occurs is given by
$$p(m)=\tr(M_m^{\dag}M_m\rho),$$ and the state of the post-measurement system
is $M_m\rho M_m^{\dag}/p(m).$

\section{\Qimp}\label{sec:qimp}
We define the syntax and operational semantics of a simple classical-quantum imperative  language called \Qimp. The language is essentially extended from \textbf{IMP} \cite{Win93} by adding quantum data and a few operations for manipulating quantum data.

\subsection{Syntax }
We assume three types of data in our language: {\tt Bool} for booleans,  {\tt Int} for integers, and qubits {\tt Qbt} for quantum data. Let $\mathbb{Z}$ be the set of constant integer numbers, ranged over by $n$. Let $\Cvar$, ranged over by $x,y,...$, be the set of classical variables, and $\Qvar$, ranged over by $q,q',...$, the set of quantum variables. It is assumed that both $\Cvar$ and $\Qvar$ are countably infinite. We assume a set {\Aexp} of arithmetic expressions over {\tt Int}, which includes $\Cvar$ as a subset and is ranged over by $a, a',...$, and a set of boolean-valued expressions $\Bexp$, ranged over by $b,b',...$, with the usual boolean constants {\true, \false} and boolean operators $\neg, \wedge,\vee$. In particular, we let $a=a'$ and $a\leq a'$ be boolean expressions for any $a,a'\in\Aexp$. We further assume that only classical variables can occur free in both arithmetic and boolean expressions. 

We let $U$ range over unitary operators, which can be  user-defined  matrices or built in if the language is implemented. For example, a concrete $U$ could be the $1$-qubit Hadamard operator $H$, or the 2-qubit controlled-NOT operator $\CNOT$, etc. Similarly, we write $M$ for the measurement described by a collection $\{M_i\}$ of measurement operators, with each index $i$ representing a measurement outcome. For example, to describe the measurement of the qubit referred to by variable $q$ in the computational basis, we can write $M:=\{M_0, M_1\}$, where $M_0=|0\>_q\<0|$ and $M_1=|1\>_q\<1|$.

\leaveout{ 
The following list summarizes the syntactic sets associated with \Imp.
	\begin{itemize}
		\item real constants $\mathbb{R}$, consisting of all real constant numbers, ranged over by \emph{metavariables} \Math{n,m}
		\item truth values \Tru=$\sset{\true,\false}$,
                \item classical variables \Cvar, ranged over by \Math{x, y}
                \item quantum variables \Qvar, ranged over by \Math{q, r}
		\item arithmetic expressions \Aexp, ranged over by \Math{a}
		\item boolean expressions \Bexp, ranged over by \Math{b}
		\item commands \Com, ranged over by \Math{c}
	\end{itemize}
} 

	Sometimes we use metavariables which are primed or subscripted, e.g. $x', x_0$ for classical variables. We abbreviate a tuple of quantum variables $\pair{q_1,...,q_n}$ as $\bar{q}$ if the length $n$ of the tuple is not important.
	The formation rules for arithmetic and boolean expressions as well as commands are defined by the following grammar.
	\begin{itemize}
		\item For \Aexp: \quad \Math{a ::= n \mid x \mid a_0+a_1 \mid a_0 - a_1 \mid a_0\times a_1}
		\item For \Bexp: \quad \Math{b ::=\true \mid \false \mid a_0=a_1 \mid a_0 \leq a_1 \mid \neg b \mid b_0\wedge b_1 \mid b_0 \vee b_1}
		\item For {\Com}: \[\begin{array}{rl}
		 c  ::= & \Skip \mid x:=a \mid c_0;c_1 \mid \If b \Then c_0 \Else c_1 \mid \While b \Do c \\
		 & \mid q := |0\> \mid \tU[\bar{q}] \mid x := \tM[\bar{q}]
		 \end{array}\]
	\end{itemize}

An arithmetic expression can be an integer, a variable, or built from other arithmetic expressions  by addition, subtraction,  or multiplication.  A boolean expression can be formed by comparing arithmetic expressions or by using the usual boolean operators. A command can be a skip statement, a classical assignment, a conditional statement, or  a while-loop, as in many classical imperative languages. In addition, there are three commands that involve quantum data. The command $q := |0\>$ initialises the qubit referred to by variable $q$ to be the basis state $|0\>$. The command $\tU[\bar{q}]$ applies the unitary operator $\tU$ to the quantum system referred to by $\bar{q}$. The command $x := \tM[\bar{q}]$ performs a measurement $M$ on $\bar{q}$ and assigns the measurement outcome to $x$. It differs from a classical assignment because the measurement $M$ may change the quantum state of $\bar{q}$, besides the fact that the value of $x$ is updated.

\subsection{Operational Semantics}
Since the execution of a {\Qimp} program may involve both classical and quantum data, we
consider the setting where the CPU that processes the program has two registers: one stores classical data and the other quantum data. Therefore, we will  model a machine state as a pair composed of a classical state and a quantum state.

The notion of classical state is standard. Formally, a \emph{classical state} is a function \Math{\state: \Cvar\rightarrow\mathbb{Z}} from classical variables to integers. Thus $\state(x)$ is the value of variable $x$ in state $\state$. The notion of quantum state is slightly more complicated. For each quantum variable $q\in \Qvar$, we assume a $2$-dimensional Hilbert space $\CH_q$ to be the state space of the $q$-system. For any finite subset $V$ of $\Qvar$, we denote
\[ \CH_V \ = \ \bigotimes_{q\in V}\CH_q.\] 
That is, $\CH_V$ is the tensor product of the individual state spaces of all the quantum variables in $V$. Throughout the paper, when we refer to a subset of $\Qvar$, it is always assumed to be finite.
 Given $V\subseteq\Qvar$, 
 the set of \emph{quantum states} consists of all partial density operators in the space $\CH_V$, denoted by $\padist{\cH_V}$. 
 A \emph{machine state} is a pair $\pair{\state,\qstate}$ where $\state$ is a classical state and $\qstate$ a quantum state. In the presence of measurements, we often need to consider an ensemble of states. For that purpose, we introduce a notion of distribution.
 \begin{definition}
 	Suppose $V\subseteq\Qvar$ and $\Sigma$ is the set of classical states, i.e., the set of functions of type $\Cvar\rightarrow\mathbb{Z}$.
 	A \emph{partial density operator valued distribution (POVD)} is a function $\povd: \Sigma\rightarrow\padist{\cH_V}$ with $\sum_{\sigma\in\Sigma}\tr(\mu(\sigma))\leq 1$.
 \end{definition}
 Intuitively, a POVD $\povd$ represents a collection of machine states where each classical state $\sigma$ is associated with a quantum state $\povd(\sigma)$. 
The notation of POVD is called  classical-quantum state in \cite{feng2020quantum}.
 If the collection has only one element $\sigma$, we explicitly write $(\sigma,\povd(\sigma))$ for $\povd$.
 The support of $\povd$, written $\support{\povd}$, is the set $\{\sigma\in\Sigma \mid \povd(\sigma)\not=0\}$.
 We can also define the addition of two distributions by letting
 $(\povd_1+\povd_2)(\sigma) = \povd_1(\sigma)+\povd_2(\sigma)$.
 
 A \emph{configuration} is a pair $\pair{e,\state,\qstate}$, where $e$ is an expression and $(\state,\qstate)$ is a POVD.
 We define the small-step operational semantics of arithmetic and boolean expressions as well as commands in a syntax-directed way  by  using an evaluation relation $\hookrightarrow$ between configurations. In Figure~\ref{fig:evalab} we list the rules for evaluating integer variables, sums, and expressions of the form $a_0\leq a_1$; the rules for other arithmetic and boolean expressions are similar.
When evaluating an arithmetic or boolean expression, we only rely on the information from the given classical state, therefore we omit the quantum state in the configuration. This is not the case when we execute commands.

\begin{figure}
\[\begin{array}{l}
\slinfer[\Rlts{Evaluation of local variables}]{\pair{x,\state}\hookrightarrow \pair{\state(x), \state}}\VS\\
\prooftree
\pair{a_0,\state}\hookrightarrow \pair{a_0',\state}
\justifies
\pair{a_0+a_1,\state}\hookrightarrow \pair{a_0'+a_1,\state}
\endprooftree\qquad
\prooftree
\pair{a_1,\state}\hookrightarrow \pair{a_1',\state}
\justifies
\pair{n+a_1,\state}\hookrightarrow \pair{n+a'_1,\state}
\endprooftree
\VS\\

\pair{n+m,\state}\hookrightarrow\pair{p,\state} \qquad \mbox{if $p$ is the sum of $n$ and $m$}\VS\\
\prooftree
\pair{a_0,\state}\hookrightarrow \pair{a'_0,\state}
\justifies
\pair{a_0\leq a_1,\state}\hookrightarrow \pair{a'_0\leq a_1,\state}
\endprooftree
		\qquad
\prooftree
\pair{a_1,\state}\hookrightarrow \pair{a'_1,\state}
\justifies
\pair{n\leq a_1,\state}\hookrightarrow \pair{n\leq a'_1,\state}
\endprooftree
\VS\\
		
\pair{n\leq m,\state}\hookrightarrow\pair{\true,\state} \qquad \mbox{if $n$ is less than or equal to  $m$.}
\VS\\
\pair{n\leq m,\state}\hookrightarrow\pair{\false,\state} \qquad \mbox{if $n$ is  greater than  $m$.}
\end{array}\]
\caption{Evaluation of arithmetic and boolean expressions (selected rules)}\label{fig:evalab}
\end{figure}


Let $\state$ be a classical state and $n\in \mathbb{Z}$. We write $\state[n/x]$ for the updated state satisfying
\[
\state[n/x](y)=\left\{\begin{array}{ll}
		n & \mbox{if $y=x$,}\\
		\state(y) & \mbox{if $y\not=x$.}
	\end{array}\right.
	\]
	
	We are going to write $\rightarrow$ for the execution of commands. 
	The transition rules are given in Figure~\ref{fig:exec}. Here we introduce a special command {\bf nil} that stands for a successful termination of programs.
We follow \cite{Yin12} to define the operational semantics of quantum measurements in a non-deterministic way, and the probabilities of different branches
are encoded in the quantum part of the configurations. For that reason we need to take partial density
operators instead of the normalised density operators to represent quantum states.
 After the measurement $M$ defined by some measurement operators $M_i$, the  original state $(\state,\qstate)$ may evolve  into a new state whose classical part is the updated state $\state[i/x]$ and  the quantum part is the new quantum state $M_i\qstate M_i^\dag$. In all other rules, the execution of a command changes a configuration to another one. Among them, the rules for initialising qubits and unitary transformations only affect the quantum part of the original machine state. On the contrary, the commands for manipulating classical data only update the classical part of a state.

\begin{figure}
\[\begin{array}{l}
\pair{\Skip,\state,\qstate}\rightarrow \pair{\nil,\state,\qstate}\VS\\
 
 \prooftree
 \pair{a,\state}\hookrightarrow \pair{a',\state'}
 \justifies
 \pair{x:=a,\state,\qstate}\rightarrow \pair{x:=a',\state',\qstate}
 \endprooftree
 \qquad
\pair{x:=n,\state,\qstate}\rightarrow \pair{\nil,\state[n/x],\qstate} \VS\\
\prooftree\pair{c_0,\state,\qstate}\rightarrow \pair{c'_0,\state',\qstate'} 
\justifies
\pair{c_0;c_1,\state,\qstate}\rightarrow \pair{c'_0;c_1,\state',\qstate'}
\endprooftree
\qquad
\prooftree
\pair{c_1,\state,\qstate}\rightarrow \pair{c'_1,\state',\qstate'} 
\justifies
\pair{\nil;c_1,\state,\qstate}\rightarrow \pair{c'_1,\state',\qstate'}
\endprooftree
\VS\\
		
\prooftree
\pair{b,\state}\hookrightarrow \pair{b',\state'}
\justifies
\pair{\If b \Then c_0 \Else c_1,\state,\qstate}\rightarrow \pair{\If b' \Then c_0 \Else c_1,\state',\qstate'}
\endprooftree
\VS\\
\pair{\If \true \Then c_0 \Else c_1,\state,\qstate}\rightarrow \pair{ c_0,\state,\qstate}
\qquad
\pair{\If \false \Then c_0 \Else c_1,\state,\qstate}\rightarrow \pair{ c_1,\state,\qstate}
\VS\\

\pair{\While b \Do c,\state,\qstate} \rightarrow \pair{\If b \Then (c;\While b\Do c)\Else \Skip,\state,\qstate}
\VS\\

\pair{q:=|0\>,\state,\qstate}\rightarrow \pair{\nil,\state,\qstate'}
\qquad\mbox{ with } \qstate'=  |0\>_q\<0|\qstate |0\>_q\<0| + |0\>_q\<1|\qstate |1\>_q\<0| 
\VS\\
\pair{\tU[\bar{q}],\state,\qstate}\rightarrow \pair{\nil,\state,\tU\qstate\tU^\dag}
\VS\\
\prooftree 
M:=\{M_i\}_{i\in I}
\justifies
\pair{x:=\tM[\bar{q}],\state,\qstate}\rightarrow \pair{\nil, \state[i/x],M_i \qstate M_i^\dag}
\endprooftree
\end{array}\]
\caption{Execution of commands}\label{fig:exec}
\end{figure}


\subsection{Denotational Semantics}

For the purpose of presenting the denotational semantics, we add an {\abort} command that halts the computation with no result.
We interpret programs as POVD transformers. 
We write 
$\PState$ for the set of POVDs called distribution states.

\begin{lemma}
We impose an order between POVDs by letting $\povd_1\leq\povd_2$ if for any classical state $\state$ we have $\povd_1(\state)\sqsubseteq\povd_2(\state)$, where $\sqsubseteq$ is the L\"{o}wner order.
Let $(\povd_n)_{n\in\Nat}\in\PState$ be an increasing sequence of POVDs. This sequence converges to some POVD $\povd_\infty$ and $\povd_n\leq \povd_\infty$ for any $n\in \Nat$.
\end{lemma}

Given an expression $e$, we denote its interpretation with respect to machine state $(\state,\qstate)$ by $\Denote{e}_{(\state,\qstate)}$. 
The denotational semantics of commands is displayed in Figure~\ref{fig:de}, where we omit the denotational semantics of arithmetic and boolean expressions such as $\Denote{a}_\state$ and $\Denote{b}_\state$, which is almost the same as in the classical setting because the quantum part  plays no role for those expressions. This is an extension of the semantics for probabilistic programs presented in \cite{BEGGHS18}. Instead of probabilistic assignments are measurements of quantum systems. A state evolves into a POVD after some quantum qubits are measured, with the measurement outcomes assigned to a classical variable. Two other quantum commands, initialisation of qubits and unitary operations, are deterministic and only affect the quantum part of a state. As usual, we define the semantics of a loop ($\While b \Do c$) as the limit of its lower approximations, where the $n$-th lower approximation of $\Denote{\While b \Do c}_{(\state,\qstate)} $ is $\Denote{(\If b \Then c)^n; \If b \Then \abort}_{(\state,\qstate)} $, where ($\If b \Then c$) is shorthand for ($\If b \Then c \Else \Skip$) and $c^n$ is the command $c$ iterated $n$ times with $c^0\equiv\Skip$. The limit exists because the sequence $(\Denote{(\If b \Then c)^n; \If b \Then \abort}_{(\state,\qstate)} )_{n\in\Nat}$ is increasing and bounded. We write $\varepsilon$ for the special POVD whose support is the empty set.

\begin{proposition}
 The semantics $\Denote{c}_{(\state,\qstate)} $ of a command $c$ in initial state ${(\state,\qstate)} $ is a POVD. The lifted semantics $\Denote{c}_\povd$ of a command $c$ in initial POVD $\povd$ is a POVD.	
\end{proposition}

\begin{figure}
\[\begin{array}{rcl}
\Denote{\Skip}_{(\state,\qstate)} & = & (\state,\qstate)  \VS\\
\Denote{\abort}_{(\state,\qstate)}  & = & \varepsilon \VS\\
\Denote{x := a}_{(\state,\qstate)}  & = & (\state[\Denote{a}_\state/x], \qstate) \VS\\
\Denote{c_0;c_1}_{(\state,\qstate)}  & = & \Denote{c_1}_{\Denote{c_0}_{(\state,\qstate)} }\VS\\
\Denote{\If b \Then c_0 \Else c_1}_{(\state,\qstate)}  & = & \left\{\begin{array}{ll}
                            \Denote{c_0}_{(\state,\qstate)}  & \mbox{if $\Denote{b}_\state=\true$}\\
                            \Denote{c_1}_{(\state,\qstate)}  & \mbox{if $\Denote{b}_\state=\false$}
                           \end{array}\right.\VS\\
\Denote{\While b \Do c}_{(\state,\qstate)}  & = &  \lim_{n\rightarrow\infty}
\Denote{(\If b \Then c)^n; \If b \Then \abort}_{(\state,\qstate)}               \VS\\
\Denote{q := |0\>}_{(\state,\qstate)}  & = &  \pair{\state,\qstate'}\\
& & \mbox{where $\qstate' :=  |0\>_q\<0|\qstate |0\>_q\<0| + |0\>_q\<1|\qstate |1\>_q\<0| $} \VS\\
\Denote{\tU[\bar{q}] }_{(\state,\qstate)}  & = & \pair{\state, \tU\qstate\tU^\dag} \VS\\
\Denote{x := \tM[\bar{q}]}_{(\state,\qstate)}  & = & \povd \VS\\
 & & \mbox{where $M=\sset{M_i}_{i\in I}$ and $\povd(\sigma')=\sum_i\{M_i \qstate M_i^\dag \mid \state[i/x]=\sigma'\}$}
\vspace{4mm}\\
\Denote{c}_{\povd} & = & \sum_{\state\in\support{\povd}} \Denote{c}_{(\state,\povd(\state))}.
\end{array}\]
\caption{Denotational semantics of commands}\label{fig:de}
\end{figure}

The operational and denotational semantics are related by the following theorem.
\begin{theorem}
  For any command $c$ and state $(\sigma,\rho)$, we have
  \[\Denote{c}_{(\sigma,\rho)} = \sum_{i}\{(\sigma_i,\rho_i) \mid \pair{c,\sigma,\rho}\rightarrow ^*\pair{\nil,\sigma_i,\rho_i}\}\ .\]
\end{theorem}
\begin{proof}
  We proceed by induction on the structure of $c$. The most difficult case is when $c\equiv \While b \Do c'$ for some command $c'$. Below we consider this case.

  Let $\textbf{While}^n = (\If b \Then c')^n; \If b \Then \abort$ and
  $\pair{c,\sigma,\rho} \rightarrow^{n}\pair{\nil,\sigma',\rho'}$ be the sequence of maximal transitions from $\pair{c,\sigma,\rho}$ such that the unfolding rule
  \[\pair{\While b \Do c',\state'',\qstate''} \rightarrow \pair{\If b \Then (c';\While b\Do c')\Else \Skip,\state'',\qstate''}\ ,\] for any $\state''$ and $\qstate''$, has been applied at most $n$ times.
  \[\mbox{\bf Claim: }\hspace{2cm} \Denote{\textbf{While}^n}_{(\sigma,\rho)} = \sum_i\{(\sigma_i,\rho_i)\mid
  \pair{c,\sigma,\rho} \rightarrow^{n+1}\pair{\nil,\sigma_i,\rho_i}\}\ .\]

  We prove the above claim by induction on $n$.
  \begin{itemize}
  \item $n=0$. On the left hand side, we have
    $$\Denote{\textbf{While}^0}_{(\sigma,\rho)} = \Denote{\If b \Then \abort}_{(\sigma,\rho)}
    = \left\{\begin{array}{ll}
                            \varepsilon\  & \mbox{if $\Denote{b}_\state=\true$}\\
                           (\state,\qstate)  & \mbox{if $\Denote{b}_\state=\false$}
    \end{array}\right. . $$
    On the right hand side, we observe that
    \begin{equation}\label{eq:while}
      \begin{array}{rcl}
    \pair{\While b \Do c',\state,\qstate} & \rightarrow & \pair{\If b \Then (c';\While b\Do c')\Else \Skip,\state,\qstate} \\
    & \rightarrow^* &  \left\{\begin{array}{ll}
                            \pair{(c';\While b\Do c'),\sigma,\rho}\  & \mbox{if $\Denote{b}_\state=\true$}\\
                           \pair{\Skip, \state,\qstate}  & \mbox{if $\Denote{b}_\state=\false$}
    \end{array}\right.
    \end{array}\end{equation}
    The unfolding rule has been used in the first reduction step in (\ref{eq:while}). 
    If $\Denote{b}_\state=\false$ then the claim clearly holds. If $\Denote{b}_\state=\true$ then the configuration $\pair{(c';\While b\Do c'),\sigma,\rho}$ cannot reduce to any $\pair{\nil, \sigma'',\rho''}$ without using the unfolding rule again, which means that there is no maximal transition from $\pair{c,\sigma,\rho}$ that uses the unfolding rule at most once. It follows that the claim also holds in this case.

  \item Suppose  $n=k+1$ and the claim holds for some $k$. On the left hand side, we have
    \begin{equation}\label{eq:m0}
    \begin{array}{rcl}
    \Denote{\textbf{While}^{k+1}}_{(\sigma,\rho)}  & = & \Denote{\If b \Then c'; \textbf{While}^k}_{(\sigma,\rho)}\\
    & = & \left\{\begin{array}{ll}
                           \Denote{\textbf{While}^k}_{\Denote{c'}_{(\sigma,\rho)}}  & \mbox{if $\Denote{b}_\state=\true$}\\
                           (\state,\qstate)  & \mbox{if $\Denote{b}_\state=\false$}
    \end{array}\right.
      \end{array}
      \end{equation}
    On the right hand side, we have the same transitions as in (\ref{eq:while}). If $\Denote{b}_\state=\false$ then the claim clearly holds. If $\Denote{b}_\state=\true$ then we infer as follows.
    Since $c'$ is a subterm of $c$, we know from the hypothesis of the structural induction that \begin{equation}\label{eq:0}
      \Denote{c'}_{(\sigma,\rho)}
      = \sum_{j\in J}\{(\sigma_j,\rho_j) \mid \pair{c',\sigma,\rho}\rightarrow ^*\pair{\nil,\sigma_j,\rho_j}\}
      \end{equation}
    for some set $J$.
    It follows that
    \begin{equation}\label{eq:aa}\Denote{\textbf{While}^k}_{\Denote{c'}_{(\sigma,\rho)}} = \sum_{j\in J}\Denote{\textbf{While}^k}_{(\sigma_j,\rho_j)} \ .
    \end{equation} By induction hypothesis on $k$,
    \begin{equation}\label{eq:bb}
      \Denote{\textbf{While}^k}_{(\sigma_j,\rho_j)} = \sum_{i\in I_j}\{(\sigma_i,\rho_i)\mid
      \pair{c,\sigma_j,\rho_j} \rightarrow^{k+1}\pair{\nil,\sigma_i,\rho_i}\}
    \end{equation}
    for some index set $I_j$.
    As a result, when $\Denote{b}_\state=\true$, we have
    \begin{equation}\label{eq:c}
      \begin{array}{rll}
        \pair{c,\state,\qstate} & \rightarrow & \pair{\If b \Then (c';c)\Else \Skip,\state,\qstate} \\
        & \rightarrow^* & \pair{(c';c),\sigma,\rho} \\
        & \rightarrow^* & \pair{c, \sigma_j,\rho_j} \qquad\mbox{by (\ref{eq:0})} \\
        & \rightarrow^{k+1} & \pair{\nil,\sigma_i,\rho_i} \qquad\mbox{by (\ref{eq:bb})}
      \end{array}
    \end{equation}
    for each $j\in J$ and $i\in I_j$. This means that
    \begin{equation}\label{eq:d}
    \pair{c,\sigma,\rho} \rightarrow^{k+2} \pair{\nil,\sigma_i,\rho_i} 
    \end{equation}
    for each $j\in J$ and $i\in I_j$. Thus, we rewrite (\ref{eq:bb}) as follows.
    \begin{equation}\label{eq:e}
      \Denote{\textbf{While}^k}_{(\sigma_j,\rho_j)} = \sum_{i\in I_j}\{(\sigma_i,\rho_i)\mid
      \pair{c,\sigma,\rho} \rightarrow^{k+2}\pair{\nil,\sigma_i,\rho_i}\}
    \end{equation}
    Combining (\ref{eq:m0}), (\ref{eq:aa}) and (\ref{eq:e}), we obtain the desired result that
    \[\Denote{\textbf{While}^{k+1}}_{(\sigma,\rho)} = \sum_{j\in J}\sum_{i\in I_j} \{(\sigma_i,\rho_i)\mid
      \pair{c,\sigma,\rho} \rightarrow^{k+2}\pair{\nil,\sigma_i,\rho_i}\}\]
  \end{itemize}
  So far we have proved the claim. Then by taking the limit on both sides of the claim, we see that $\Denote{c}_{(\sigma,\rho)} = \sum_{i}\{(\sigma_i,\rho_i) \mid \pair{c,\sigma,\rho}\rightarrow ^*\pair{\nil,\sigma_i,\rho_i}\}$.
\end{proof}

\section{An Abstract Proof System}\label{sec:abs}
In this section, we present an abstract proof system, where assertions are arbitrary predicates on POVDs. We show that the proof system is sound and relatively complete.
\begin{definition}
	The set $\Assn$ of assertions is defined as $\cal{P}(\PState)$, the powerset of $\PState$. Each assertion $P$ can be constructed by the following grammar.
	\[P ~:=~ \Char_{\povd} \mid S \mid  \neg P \mid P_1\wedge P_2 \mid \Box\psi  \mid P_1\oplus P_2 \mid P[f]\]
	where $\mu\in\PState$, $S\subseteq \PState$, $\psi$ is a predicate over states and $f$ is a function from $\PState$ to $\PState$.
	
\end{definition}
Here $\Char_\povd$ is also called the characteristic function of the 
 POVD $\povd$, which is a predicate requiring that
$\Char_\povd$ holds on $\povd'$ if and only if $\povd'=\povd$, for any distribution state $\povd'$.
The satisfaction relation $\models$ between a POVD and an assertion is defined as follows.
\[\begin{array}{rcl}
\povd \models \Char_{\povd'} & \mbox{iff} & \povd=\povd' \\
\povd \models S & \mbox{iff} & \povd\in S \\
\povd \models \neg P & \mbox{iff} & 
\mbox{not } \povd\models P \\
\povd \models P_1\wedge P_2  & \mbox{iff} & 
\povd \models P_1 \wedge \povd\models P_2\\
\povd \models \Box\psi & \mbox{iff} &  \forall \state.\ \state\in\support{\povd} \Rightarrow \Denote{\psi}_\state =\true\\
\povd \models P_1\oplus P_2 & \mbox{iff} & \exists \povd_1, \povd_2.\ \povd=\povd_1 + \povd_2 \wedge \povd_1\models P_1 \wedge \povd_2\models P_2\\
\povd \models P[f] & \mbox{iff} & f(\povd)\models P
\end{array}\]
Let $\Denote{P}:=\{\povd \mid \povd\models P\}$ be the semantic interpretation of assertion $P$. We see that boolean operations of assertions are represented by set operations.  For example, we have $\Denote{\neg P}=\mathcal{P}(\PState)\backslash \Denote{P}$ and $\Denote{P_1\wedge P_2}=\Denote{P_1}\cap\Denote{P_2}$.
The predicate $\Box\psi$ is lifted from a state predicate by requiring that $\Box\psi$ holds on the POVD $\povd$ when $\psi$ holds on all the states in the support of $\povd$.  For example, a particular predicate over states is a boolean expression $b$ with $\state\models b$ iff $\Denote{b}_\state=\true$. Therefore, the predicate $\Box b$ holds on the POVD $\povd$ when $b$ evaluates to be true under any state $\state$ in the support of $\povd$.
The assertion $P_1\oplus P_2$ holds on the POVD $\povd$ if we can split $\povd$ into the sum of two POVDs such that $P_1$ and $P_2$ hold on each of them. Lastly, $P[f]$ holds on a POVD $\povd$ only when $P$ holds on the image of $\povd$ under $f$.

\leaveout{ 
We define boolean operations of assertions by set operations. For example, $P\wedge P' := P\cap P'$ and $\neg P := \PState\backslash P$.
Given a predicate $\psi$ over states, we lift it to be a predicate $\Box\psi$ over subdistributions by letting
\[ \Box\psi (\Delta) ~:=~ \forall \state. \state\in\support{\Delta} \Rightarrow \psi(\state) .\]
Intuitively, $\Box\psi$ holds on the subdistribution $\Delta$ when $\psi$ holds on all the states in the support of $\Delta$.  For example, a particular predicate over states is a boolean expression $b$ with $m\models b$ iff $\Denote{b}_m=\true$. Therefore, the predicate $\Box b$ holds on the subdistribution $\Delta$ when $b$ evaluates to be true under any state $m$ in the support of $\Delta$.
To reason about branching commands, we combine two assertions $P_1$ and $P_2$ into a new assertion $P_1\oplus P_2$ by letting
\[(P_1\oplus P_2)(\Delta) ~:=~ \exists \Delta_1, \Delta_2. \Delta=\Delta_1 + \Delta_2 \wedge P_1(\Delta_1) \wedge P_2(\Delta_2) .\]
In other words, the assertion $P_1\oplus P_2$ holds on the subdistribution $\Delta$ if we can split $\Delta$ into the sum of two subdistributions such that $P_1$ and $P_2$ holds on each of them.
Given an assertion $P$ and a function $f$ from $\PState$ to $\PState$, we define $$P[f] := \lambda \Delta. P(f(\Delta)) .$$ Therefore, $P[f]$ holds on a subdistribution $\Delta$ only when $P$ holds on the image of $\Delta$ under $f$.
} 

\begin{definition}
A sequence of assertions $(P_n)_{n\in\Nat^\infty}$ is \emph{u-closed}, if for each increasing sequence of POVDs $(\povd_n)_{n\in\Nat}$ such that $\povd_n\models P_n$ for all $n\in\Nat$, we have $\lim_{n\rightarrow\infty}\povd_n\models P_\infty$.
\end{definition}

\begin{definition}
	A judgement is a triple in the form $\triple{P}{c}{P'}$, where $c$ is a command, $P$ and $P'$ are assertions. It is valid, written $\models\triple{P}{c}{P'}$, if 
	\[\forall \povd.\ \povd\models P \ \Rightarrow \ \Denote{c}_\povd \models P' .\]
\end{definition}

In Figure~\ref{fig:proofrules} we give the rules for an abstract proof system denoted by 
 $\mathcal{S}_a$. It extends the system in \cite{BEGGHS18} with the last three rules to handle the manipulations of quantum systems. In order to show the soundness of $\mathcal{S}_a$, we need a few technical lemmas.

\begin{figure}
\[\begin{array}{c}
\linfer{}{\triple{P}{\Skip}{P}}{[\sf Skip]}
\qquad
\linfer{}{\triple{P}{\abort}{\Box\false}}{[\sf Abort] }
\qquad
\linfer{}{\triple{P[\Denote{x:=a}]}{x : = a }{P}}{[\sf Assgn]}
\VS\\
\prooftree 
\triple{P_0}{c_0}{P_1} \quad
\triple{P_1}{c_1}{P_2}
\justifies
\triple{P_0}{c_0;c_1}{P_2}
\using
[\sf Seq]
\endprooftree

\qquad
\prooftree
\triple{P_0}{c}{P'_0}
\quad
\triple{P_1}{c}{P'_1}
\justifies
\triple{P_0\oplus P_1}{c}{P'_0\oplus P'_1}
\using
[\sf Split]
\endprooftree
\VS\\

\prooftree
\triple{P_0\wedge\Box b}{c_0}{P'_0}
\quad
\triple{P_1\wedge\Box \neg b}{c_1}{P'_1}
\justifies
\triple{(P_0\wedge\Box b)\oplus (P_1 \wedge \Box\neg b)}
{\If b \Then c_0 \Else c_1}{P'_0\oplus P'_1}
\using
[\sf Cond]
\endprooftree
\VS\\
\prooftree
\justifies
\triple{\false}{c}{P}
\using
[\sf Absurd]
\endprooftree

\qquad
\prooftree
P_0 \Rightarrow P_1 \quad
\triple{P_1}{c}{P_2} \quad
P_2 \Rightarrow P_3
\justifies
\triple{P_0}{c}{P_3}
\using
[\sf Conseq]
\endprooftree

\qquad
\prooftree
\forall \povd.\ \triple{\Char_\povd\wedge P}{c}{P'}
\justifies
\triple{P}{c}{P'}
\using
[\sf All]
\endprooftree
\VS\\

\prooftree
\begin{array}{c}
\uclosed((P'_n)_{n\in\Nat^\infty})\\
\forall n.\ 
\triple{P_n}{\If b \Then c}{P_{n+1}}
\quad
\forall n.\ \triple{P_n}{\If b \Then \abort}{P'_n}
\end{array}
\justifies
\triple{P_0}{\While b \Do c}{P'_\infty \wedge \Box \neg b}
\using
[\sf While]
\endprooftree
\VS\\

\linfer{}{\triple{P[\Denote{q := |0\>}]}{q := |0\> }{P}}{[\sf QInit]}
\qquad
\linfer{}{\triple{P[\Denote{\tU[\bar{q}]}]}{\tU[\bar{q}]}{P}}{[\sf QUnit]}
\VS\\

\linfer{}{\triple{P[\Denote{x := \tM[\bar{q}]}]}{x := \tM[\bar{q}]}{P}}{[\sf QMeas]}
\end{array}\]
\caption{Proof rules for $\mathcal{S}_a$}\label{fig:proofrules}
\end{figure}

\begin{lemma}\label{lem:assgn1}
Let $P$ be an assertion and $c$ a command. Then 
$\models \triple{P[\Denote{c}]}{c}{P}$.
\end{lemma}
\begin{proof}
	Suppose $\povd$ is a distribution state and $\povd\models P[\Denote{c}]$. By the definition of $P[\Denote{c}]$, this means that $\Denote{c}_\povd \models P$, which is the desired result.
\end{proof}

\begin{lemma}\label{lem:split}
	Let $\state$ be a classical state, $\qstate_1, \qstate_2$ be two quantum states, and $\povd_1, \povd_2$ be two POVDs. For any command $c$, we have
	\begin{enumerate}
		\item $\Denote{c}_{(\state,\qstate_1+\qstate_2)} = \Denote{c}_{(\state,\qstate_1)} + \Denote{c}_{(\state,\qstate_2)}$;
		\item $\Denote{c}_{(\povd_1+\povd_2)} = \Denote{c}_{\povd_1} + \Denote{c}_{\povd_2}$.
	\end{enumerate}
\end{lemma}
\begin{proof}
The two clauses can be proved by a simultaneous induction on the structure of command $c$.
\end{proof}

\begin{lemma}\label{lem:seq1}
For any commands $c_0, c_1$ and distribution state $\povd$, we have	$\Denote{c_0;c_1}_\povd = \Denote{c_1}_{\Denote{c_0}_\povd}$.
\end{lemma}
\begin{proof}
	\[\begin{array}{rcl}
 	\Denote{c_1}_{\Denote{c_0}_\povd} 
& = & \sum_\state  \Denote{c_1}_{(\state,\Denote{c_0}_\povd(\state))}\\
& = & \sum_\state  \Denote{c_1}_{(\state,\sum_{\state'}\Denote{c_0}_{(\state',\povd(\state'))}(\state))}\\
& = & \sum_\state\sum_{\state'}\Denote{c_1}_{(\state,\Denote{c_0}_{(\state',\povd(\state'))}(\state))}\qquad\mbox{by Lemma~\ref{lem:split}(1)}\\
& = & \sum_{\state'}\sum_{\state}\Denote{c_1}_{(\state,\Denote{c_0}_{(\state',\povd(\state'))}(\state))}\\
& = & \sum_{\state'}\Denote{c_1}_{\Denote{c_0}_{(\state',\povd(\state'))}}\\
& = & \sum_{\state'}\Denote{c_0;c_1}_{(\sigma',\povd(\sigma'))} \\
& = & \Denote{c_0;c_1}_{\povd}
	\end{array}\]
\end{proof}

\begin{theorem}\label{thm:sounda}
	{\bf (Soundness)}
	Every judgement provable using the proof system $\mathcal{S}_a$ is valid.
\end{theorem} 
\begin{proof}

We analyze the cases one by one.
\begin{itemize}
\item Rule {[\sf Skip]}. Suppose $\povd\models P$ for some distribution state $\povd$. Then we have $\Denote{\Skip}_\povd = \povd$ and thus
$\Denote{\Skip}_\povd\models P$ as required.

\item Rule  {[\sf Abort]}. This case is easy by noting that $\Denote{\abort}_\povd=\varepsilon$ and $\varepsilon\models\Box\false$ for any $\povd$.

\item The cases for rules {[\sf Assgn], [\sf QInit], [\sf QUnit]}, and {[\sf QMeas]} follow from Lemma~\ref{lem:assgn1}. 

\item Rule {[\sf Seq]}. Suppose $\povd\models P_0$ for some distribution state $\povd$. By  the premises, both $\triple{P_0}{c_0}{P_1}$ and $ \triple{P_1}{c_1}{P_2}$ are valid. It follows that $\Denote{c_0}_\povd \models P_1$ and then $\Denote{c_1}_{\Denote{c_0}_\povd}\models P_2$, which is $\Denote{c_0;c_1}_\povd\models P_2$ by Lemma~\ref{lem:seq1} as required.

\item Rule {[\sf Split]}. Suppose $\povd \models P_0\oplus P_1$ for some distribution state $\povd$. Then there exist $\povd_0$ and $\povd_1$ such that $\povd=\povd_0+\povd_1$, $\povd_0\models P_0$ and $\povd_1\models P_1$. By  the premises, both $\triple{P_0}{c_0}{P'_0}$ and $ \triple{P_1}{c_1}{P'_1}$ are valid. Therefore, we have that $\Denote{c}_{\povd_0}\models P'_0$ and  $\Denote{c}_{\povd_1}\models P'_1$. By Lemma~\ref{lem:split} we obtain $\Denote{c}_\povd = \Denote{c}_{\povd_0} + \Denote{c}_{\povd_1}$. It follows that $\Denote{c}_\povd\models P'_0\oplus P'_1$ as required.

\item Rule {[\sf Cond]}. We first claim that $\triple{P_0\wedge\Box b}{\If b \Then c_0 \Else c_1}{P'_0}$ is valid. To see this, suppose $\povd$ is a POVD with $\povd\models P_0\wedge\Box b$. Obviously, we have $\povd\models\Box b$ and thus $\Denote{b}_\state=\true$ for each $\state\in\support{\povd}$. It follows that
\[\begin{array}{rcl}
\Denote{\If b \Then c_0 \Else c_1}_\povd & = & \sum_\state \Denote{\If b \Then c_0 \Else c_1}_{(\state,\povd(\state))} \\
& = & \sum_\state  \Denote{c_0}_{(\state,\povd(\state))} \\
& = & \Denote{c_0}_\povd .
\end{array}\]
By the first premise, $\triple{P_0\wedge\Box b}{c_0}{P'_0}$ is valid. Therefore, 
we have $\Denote{c_0}_\povd\models P'_0$, and thus
 $$\Denote{\If b \Then c_0 \Else c_1}_\povd \models P'_0 $$
 and the above claim is proved. Similarly, we can prove that
 $\triple{P_1\wedge\Box \neg b}{\If b \Then c_0 \Else c_1}{P'_1}$ is valid.
 By the soundness of  {[\sf Split]}, it follows that $$\triple{(P_0\wedge\Box b)\oplus (P_1 \wedge \Box\neg b)}
 {\If b \Then c_0 \Else c_1}{P'_0\oplus P'_1}$$ is also valid.
 
 \item Rule {[\sf Absurd]}. There exists no $\povd$ with $\povd\models \false$. Thus, we always have $\forall \povd.\ \povd\models \false \Rightarrow \Denote{c}_\povd\models P$. 
 
 \item Rule {[\sf Conseq]}. Let $\povd$ be a distribution state and $\povd\models P_0$. The first premise gives $\povd\models P_1$. The second premise tells us that $\Denote{c}_\povd\models P_2$. By the third premise, we derive that $\Denote{c}_\povd\models P_3$. It follows that $\triple{P_0}{c}{P_3}$ is valid.
 
 \item Rule {[\sf All]}. Let $\povd$ be a POVD and $\povd\models P$. It is clear that $\povd\models \Char_\povd \wedge P$. By the premise, $\triple{\Char_\povd\wedge P}{c}{P'}$ is valid. Therefore, we have $\Denote{c}_\povd \models P'$, and thus $\triple{P}{c}{P'}$ is valid.
 
 \item Rule {[\sf While]}. We first observe that, for any state $(\state,\qstate)$,
 \[ \Denote{\If b \Then \abort}_{(\state,\qstate)} = \left\{\begin{array}{ll}
  	\varepsilon & \mbox{if $\Denote{b}_{(\state,\qstate)}=\true$}\\
  	{(\state,\qstate)} & \mbox{if $\Denote{b}_{(\state,\qstate)}=\false$} \ .
  	\end{array}
  \right. \]
  Thus, if a state $\state'$ is in the support of $\Denote{\If b \Then \abort}_{(\state,\qstate)} $, it must be the case that $\state'\models \neg b$. Furthermore, for any distribution state $\povd$, if a state $\state'$ is in the support of $\Denote{\If b \Then \abort}_\povd$ then $\state'\models \neg b$. It follows that, for any command $c'$ and distribution state $\povd$, we have 
  \[\Denote{c'; \If b \Then \abort}_\povd\models \Box \neg b. \]
  By definition, $\Denote{\While b \Do c}_\povd$ is the limit of the sequence 
   \[(\Denote{(\If b \Then c)^n; \If b \Then \abort}_\povd)_{n\in\Nat}\]
   and  so we have that 
   \begin{equation}\label{e:1}
   \Denote{\While b \Do c}_\povd\models \Box \neg b.
   \end{equation}
   
   By the first premise and the soundness of {[\sf Seq]}, it is easy to show by induction that
   \[\forall n.\ 
   \triple{P_0}{(\If b \Then c)^n}{P_{n}}\]
   is valid. By the second premise and {[\sf Seg]} again, the following judgement
   \[\forall n.\ 
   \triple{P_0}{(\If b \Then c)^n; \If b \Then \abort}{P'_{n}}\]
   is valid. Let $\povd$ be any POVD with $\povd\models P_0$. Then
     \[\forall n.\ 
   \Denote{(\If b \Then c)^n; \If b \Then \abort}_\povd\models P'_{n} .\]
   By assumption, the sequence of assertions $(P'_n)_{n\in\Nat^\infty}$ is u-closed. Hence, we can infer that $$\Denote{\While b \Do c}_\povd \models P'_\infty, $$
  which means that the judgement
  \begin{equation}\label{e:2}
  \triple{P_0}{\While b \Do c}{P'_\infty}
  \end{equation} is valid.
  Combining (\ref{e:1}) and (\ref{e:2}), we finally obtain that 
  $ \triple{P_0}{\While b \Do c}{P'_\infty \wedge \Box \neg b}$ is valid.
\end{itemize}
\end{proof}

Now we turn to the relative completeness of the proof system $\mathcal{S}_a$. Formulas of the form $\Char_\povd$  will be helpful for that purpose.

\begin{lemma}\label{lem:seq}
For any distribution state $\povd$ and command $c$, 
\[\Char_\povd \Rightarrow \Char_{\Denote{c}_\povd}[\Denote{c}] .\]
\end{lemma}
\begin{proof}
Let $\povd'$ be any distribution state. 
\[\begin{array}{rcl}
\povd'\models \Char_\povd & \Leftrightarrow & \povd'=\povd \\
& \Rightarrow & \Denote{c}_{\povd'} = \Denote{c}_\povd \\
& \Leftrightarrow &  \Denote{c}_{\povd'} \models \Char_{ \Denote{c}_{\povd} } \\
& \Leftrightarrow & \povd' \models \Char_{ \Denote{c}_{\povd}}[\Denote{c}]
\end{array}\]
\end{proof}

\begin{definition}
Let $\povd$ be a distribution state and $b$ a boolean expression. The restriction $\povd_{|b}$ of $\povd$ to $b$ is the distribution state such that $\povd_{|b}(\state)=\povd(\state)$ if $\Denote{b}_\state=\true$ and $0$ otherwise.
\end{definition}
According to the definition above, it is easy to see that we can split any $\povd$ into two parts w.r.t. a boolean expression.
\begin{lemma}\label{lem:decomp}
	For any distribution state $\povd$ and boolean expression b, we have
	$\povd = \povd_{|b} + \povd_{|\neg b}$.
\end{lemma}
\begin{proof}
	This is straightforward because, at each state $\sigma$ in the support of $\povd$, the boolean expression $b$ evaluates to either $\true$ or $\false$.
\end{proof}

With Lemmas~\ref{lem:decomp} and \ref{lem:split}, it is easy to see that the denotational semantics of conditional commands can be rewritten as follows.
\begin{equation}\label{eq:cond}
\Denote{\If b \Then c_0 \Else c_1}_\povd ~=~ \Denote{c_0}_{\povd_{|b}} + \Denote{c_1}_{\povd_{|\neg b}}
\end{equation}
The following facts are also easy to show.
\begin{equation}\label{eq:Char}
\begin{array}{rcl}
 \Char_{\povd | b} & \Leftrightarrow &\Char_{\povd | b} \wedge \Box b \\
  \Char_{\povd} & \Leftrightarrow &\Char_{\povd} \wedge P \qquad\mbox{if $\povd\models P$}\\
 \Char_{\povd} & \Leftrightarrow & \Char_{\povd|b} \oplus \Char_{\povd|\neg b}\\
 \Char_{\povd_1+\povd_2} & \Leftrightarrow & \Char_{\povd_1} \oplus \Char_{\povd_2}
\end{array}
\end{equation}

\begin{lemma}\label{lem:CharP}
	For any POVD $\povd$, the following judgement is provable:
	\[\triple{\Char_{\povd}}{c}{\Char_{\Denote{c}_\povd}} . \]
\end{lemma}
\begin{proof}
We proceed by induction on the structure of $c$.
\begin{itemize}
	\item $c\equiv \Skip$. This case is immediate as $\Denote{\Skip}_\povd=\povd$ and by {[\sf Skip]} we have $\vdash \triple{\Char_{\povd}}{c}{\Char_{\povd}} . $
	\item $c\equiv\abort$. Then $\Denote{c}_\povd=\varepsilon$. For any POVD $\povd'$, we note that \[\povd'\models\Box \false ~\Leftrightarrow~ \povd'=\varepsilon ~\Leftrightarrow~ \povd'\models\Char_{\varepsilon} . \]
	By rules {[\sf Abort]} and {[\sf Conseq]} we can infer 
	$\vdash \triple{\Char_{\povd}}{c}{\Char_{\varepsilon}} $.
	
	\item $c\equiv x:=a,\ q := |0\>,\ \tU[\bar{q}]$ or $x := \tM[\bar{q}]$.  By the corresponding rules
	{[\sf Assgn]}, {[\sf QInit]}, {[\sf QUnit]} or {[\sf QMeas]}, we have 
	$$\vdash \triple{\Char_{\Denote{c}_\povd} [\Denote{c}]}{c}{\Char_{\Denote{c}_\povd}} . $$
	By Lemma~\ref{lem:seq} and rule {[\sf Conseq]}, we obtain that 
		$\vdash \triple{\Char_{\povd}} {c}{\Char_{\Denote{c}_\povd}} . $
		
	\item $c\equiv c_0;c_1$. By induction, we have $\vdash\triple{\Char_{\povd}}{c_0}{\Char_{\Denote{c_0}_\povd}}$ and 
	$\vdash\triple{\Char_{\Denote{c_0}_\povd}}{c_1}{\Char_{\Denote{c_1}_{\Denote{c_0}_\povd}}}$. Using the rule {[\sf Seq]}, we obtain that 
	$\vdash\triple{\Char_{\povd}}{c}{\Char_{\Denote{c_1}_{\Denote{c}_\povd}}}$. 
	
	\item $c\equiv \If b \Then c_0 \Else c_1$. By induction, we have 
	$\vdash \triple{\Char_{\povd_{|b}}}{c_0}{\Char_{\Denote{c_0}_{\povd_{|b}}}}$. By the first clause in (\ref{eq:Char}) and rule {[\sf Conseq]}, we have
	$\vdash \triple{\Char_{\povd_{|b}}\wedge \Box b}{c_0}{\Char_{\Denote{c_0}_{\povd_{|b}}}}$. Similarly, 
	$\vdash \triple{\Char_{\povd_{|\neg b}}\wedge \Box \neg b}{c_1}{\Char_{\Denote{c_1}_{\povd_{|\neg b}}}}$. Using rule {[\sf Cond]}, we infer
	$$\vdash \triple{(\Char_{\povd_{|b}} \wedge \Box b) \oplus (\Char_{\povd_{|\neg b}}\wedge \Box \neg b)}{\If b \Then c_0 \Else c_1}{\Char_{\Denote{c_0}_{\povd_{|b}}} \oplus \Char_{\Denote{c_1}_{\povd_{|\neg b}}}}. $$
	Using (\ref{eq:cond}), (\ref{eq:Char}) and rule {[\sf Conseq]}, we finally obtain that
	$$\vdash \triple{\Char_{\povd}}{\If b \Then c_0 \Else c_1}{\Char_{\Denote{\If b \Then c_0 \Else c_1}_\povd}}.$$
	
	\item $c\equiv \While b \Do c'$. For each $n\in\Nat$, let
	\[\begin{array}{rcl}
	P_n & = & \Char_{\Denote{(\If b \Then c')^n}_\povd} \\
	P'_n & = & \Char_{\Denote{(\If b \Then c')^n; \If b \Then \abort}_\povd} \\
	P'_\infty & = & \Char_{\lim_{n\rightarrow\infty}\Denote{(\If b \Then c')^n; \If b \Then \abort}_\povd}
	\end{array}\]
	Obviously, the sequence of assertions $(P'_n)_{n\in\Nat^\infty}$ is u-closed.
	As in the last case, we can show that $\vdash \triple{P_n}{\If b \Then c'}{P_{n+1}}$ by induction hypothesis and rules {[\sf Conseq]}, {[\sf Skip]} and {[\sf Cond]}. It is also easy to see that
	 $\vdash\triple{P_n}{\If b \Then \abort}{P'_n}$ for each $n\in \Nat$. Therefore, we can use rule {[\sf While]} to infer that 
	 $\vdash \triple{P_0}{\While b \Do c}{P'_\infty \wedge \Box \neg b}$. Using (\ref{e:1}), the second clause of (\ref{eq:Char}), and rule {[\sf Conseq]}, we obtain that $\vdash \triple{P_0}{\While b \Do c}{P'_\infty}$, which is exactly $\vdash\triple{\Char_\povd}{\While b \Do c'}{\Char_{\Denote{\While b \Do c'}_\povd}}$.
\end{itemize}
\end{proof}

With the preparations above, we are in the position to show that the proof system ${\cal{S}}_a$ is relatively complete.
\begin{theorem}
	{\bf (Relative completeness)}
	Every valid judgement is derivable in ${\cal{S}}_a$.	
\end{theorem}
\begin{proof}
Let $\triple{P}{c}{P'}$ be a valid judgement. Suppose $\povd$ be any POVD. There are two possibilities:
\begin{itemize}
	\item $\povd\models P$. The validity of the judgement says that $\Denote{c}_\povd\models P'$.  
By Lemma~\ref{lem:CharP}, we have that 
$\vdash \triple{\Char_\povd}{c}{\Char_{\Denote{c}_\povd}}$. 
By the second clause of (\ref{eq:Char}) and rule {[\sf Conseq]}, we can obtain that
$\vdash \triple{\Char_\povd\wedge P}{c}{\Char_{\Denote{c}_\povd} \wedge P'}$.  
Using {[\sf Conseq]} again gives $\vdash \triple{\Char_\povd\wedge P}{c}{ P'}$.  
 
\item $\povd\not\models P$. Then it is obvious that $\Char_\povd\wedge P \Leftrightarrow \false$. By rules {[\sf Absurd]} and {[\sf Conseq]}, we also obtain 
$\vdash \triple{\Char_\povd\wedge P}{c}{ P'}$.
\end{itemize}
Since $\povd$ is arbitrarily chosen, the premise of rule {[\sf All]} is derivable. Therefore, we can use that rule to obtain $\vdash\triple{P}{c}{P'}$.
\end{proof}

\section{A Concrete Program Logic}\label{sec:conc}
In this section, we present a concrete program logic. We first define the concrete syntax of assertions. Following \cite{BEGGHS18}, we define a two-level assertion language in Figure~\ref{fig:syn}. Formally, assertions are divided into two categories: \emph{state assertions} are formulas that describe the properties of machine states and \emph{distribution assertions} are used to describe the properties of POVDs. Distribution assertions are based on comparison of distribution expressions,  or built with first-order quantifiers and connectives, as well as  the connective $\oplus$ mentioned in Section~\ref{sec:abs}. 
A \emph{distribution expression} is either the expectation $\Exp[e]$ of a state expression $e$, the expectation $\Exp_{\bar{x}\sim M[\bar{q}]}[e]$ of state expression $e$ w.r.t. the measurement $M$, or an operator applied to distribution expressions. A \emph{state expression} is either a classical variable, the characteristic function $\Char_\psi$ of a state assertion $\psi$,  or an operator applied to state expressions. 
 Finally, a \emph{state assertion} is either a comparison of state expressions, or a first-order formula over state assertions. In particular, the boolean expressions in \textbf{Bexp} are included as state assertions. 
 Note that the set of operators is left unspecified but we assume that some basic operators such as addition, subtraction, scalar multiplication for both arithmetic expressions and matrix representation of partial density operators are included. With a slight abuse of notation, when $\leq$ is used to compare matrices, we essentially mean $\sqsubseteq$. Similarly for $<$ and $=$.

 For convenience of presentation,
 in this section we consider a general form of quantum measurement. 
\begin{definition}
	A \emph{general measurement} $M$ is a pair $\pair{\{M_i\}_{i\in I}, \Lab}$, where each $M_i$ is a measurement operator as usual, and $\Lab: I\mapsto J$ is a labelling function that maps each measurement outcome $i$ to some some label $\Lab(i)$. 
	
	If the state of a quantum system is specified by density operator $\rho$ immediately before the measurement $M$, then the probability with which those results with label $j$ occur is given by 
	\[p(j)=\sum_{i:\Lab(i)=j} \tr(M^\dag_iM_i\rho),\]
	and the state of the system after the measurement is
	\[\\\frac{\sum_{i:\Lab(i)=j} M_i\rho M^\dag_i}{p(j)} .\]
\end{definition}
General measurements are convenient to describe the situation where we would like to group some measurement outcomes. For example, if $i_1, i_2 \in I$ are two different outcomes, but for some reasons we would not like to distinguish them, then we simply give them the same label by letting $\Lab(l_1)=\Lab(l_2)$. In the special case that $l$ is the identity function $\Id$, then the labelling function has no effect and we degenerate to the usual notion of measurements.

\begin{figure}
\[\begin{array}{rcll}
e & ::= & x 
\mid \Char_{\psi} \mid o(\mathbf{e})
 \qquad
  & \mbox{(State expressions)} \\
\psi & ::= & e \comp e \mid FO(\psi) & \mbox{(State assertions)} \\
r & ::= & \Exp[e] \mid  \Exp_{\bar{x}\sim M[\bar{q}]}[e] \mid o(\mathbf{r}) & \mbox{(Distribution expressions)} \\
P & ::= &  r \comp r \mid P \oplus P \mid FO(P) & \mbox{(Distribution assertions)} \\
\rhd\!\!\lhd & \in  & \{=,\ <,\ \leq\} \qquad o\in Ops & \mbox{(Operations)}
\end{array}\]
\caption{Syntax of assertions 
}\label{fig:syn}
\end{figure}

\begin{figure}
	\[\begin{array}{rcl}
	\Denote{x}_\state & := & \state(x) \\
	\Denote{\Char_{\psi}}_\state & := & \Char_{\Denote{\psi}_{\state}} \\
	\Denote{o(e)}_\state & := & o(\Denote{e}_\state)\\
	\hline \\
	\Denote{e_1 \comp e_2}_\state & := & \Denote{e_1}_\state \comp  \Denote{e_2}_\state \\
	\Denote{FO(\psi)}_\state & := & FO(\Denote{\psi}_\state) \\
	\hline\\
	\Denote{\Exp[e]}_\povd & := & \sum_{\state}\povd(\state)\cdot \Denote{e}_\state\\
		\Denote{\Exp_{\bar{x}\sim M[\bar{q}]}[e]}_\povd & := & \sum_\state \sum_{i} M_i \povd(\state) M_i^\dag \cdot \Denote{e}_{\state[l(i)/\bar{x}]} \\
	& & \mbox{where $M=\pair{\sset{M_i}_{i\in I},\Lab}$}\\
   \Denote{o(r)}_\povd & := & o(\Denote{r}_\povd)\\
   \hline \\
   \Denote{r_1 \comp r_2}_\povd & := &  \Denote{r_1}_\povd \comp \Denote{r_2}_\povd  \\
   \Denote{P_1\oplus P_2}_\povd & := & \exists \povd_1, \povd_2.\ \povd = \povd_1 + \povd_2 \wedge \Denote{P_1}_{\povd_1} 
              \wedge\Denote{P_2}_{\povd_2}	\\
   \Denote{FO(P)}_\povd & := & FO(\Denote{P}_\povd)
	\end{array}\]
\caption{Semantics of assertions}\label{fig:sem}
\end{figure}

The interpretation of assertions is given in Figure~\ref{fig:sem}. Comparing the interpretation with that in \cite{BEGGHS18}, we see that the main difference is the introduction of a distribution expression related to a quantum measurement. The meaning of 
$\Exp_{\bar{x}\sim M[\bar{q}]}[e]$ is the expected Hermitian operator weighted by the value of $e$ after a measurement entailed by $M$. 

Note that the formula $\Box\psi$, where $\psi$ is a state assertion, can now be viewed as a syntactic sugar in view of the following lemma.
\begin{lemma}\label{lem:box}
\begin{enumerate}
\item	$\Box\psi \ \Leftrightarrow\ \Exp[\Char_{\psi}] = \Exp[\Char_{\true}]$
\item $\Box\psi \ \Leftrightarrow\ \Exp_{\bar{x}\sim M[\bar{q}]}[\Char_{\psi}] = \Exp_{\bar{x}\sim M[\bar{q}]}[\Char_{\true}]$
\item  $\Box\psi\ \Leftrightarrow\ \Box(\psi\wedge b) \oplus \Box(\psi\wedge\neg b) $
\end{enumerate}
\end{lemma}
\begin{proof}
Let us consider the first clause; the second one is similar and the third one is easier.
	\[\begin{array}{lcl}
	\povd\models\Box\psi &
	\mbox{iff} & \forall \state\in\support{\povd}. \Denote{\psi}_\state=\true\\
&	\mbox{iff} & \sum_{\state}\povd(\state)\cdot\Denote{\Char_{\psi}}_\state = \sum_{\state}\povd(\state)\cdot\Denote{\Char_{\true}}_\state\\
&	\mbox{iff} & \Denote{\Exp[\Char_{\psi}]}_\povd = \Denote{\Exp[\Char_\true]}_\povd \\
&	\mbox{iff} & \povd\models (\Exp[\Char_{\psi}] = \Exp[\Char_{\true}])
	\end{array}\]
\end{proof}

Using the concrete syntax for assertions, we propose a syntactic version of the existing proof rules  by avoiding the semantics of commands. We call the concrete proof system $\mathcal{S}_c$. Specifically, we keep all proof rules in Figure~\ref{fig:proofrules} but replace [{\sf Assgn}], [{\sf QInit}], [{\sf QUnit}], and [{\sf QMeas}] with the four rules in Figure~\ref{fig:syntacticproofrules}. 

\begin{figure}
	\[\begin{array}{c}
	
	\linfer{}{\triple{P[a/x]}{x : = a }{P}}{[\sf Assgn']}
	\qquad
	\prooftree
	\justifies
	\triple{h(P)}{q:=|0\>}{P}
	\using
	[\sf QInit']
	\endprooftree
	\VS\\
	\prooftree
	\justifies
	\triple{g^U(P)}{U[\bar{q}]}{P}
	\using
	[\sf QUnit']
	\endprooftree		
	\qquad
	\prooftree
	\justifies
	\triple{f_{x,\bar{q}}^M(P)}{x :=M[\bar{q}]}{P}
	\using
	[\sf QMeas']
	\endprooftree
	\end{array}\]
	\caption{Selected syntactic proof rules}\label{fig:syntacticproofrules}
\end{figure}

In rule [{\sf QInit'}] we use the notation $h(P)$ for a syntactic substitution. It changes all 
$\Exp_{\bar{x}\sim M[\bar{q}]}[e]$ in $P$ into $\Exp_{\bar{x}\sim M'[\bar{q}]}[e]$ and distributes over most other syntactic constructors of assertions, where
$M'$ is obtained from $M=\pair{\{M_i\}_i, l}$ by constructing two measurement operators $M_{i0}, M_{i1}$ for each $M_i$ in $M$ with the mapping $\Lab'$ given by $\Lab'(i0)=\Lab'(i1)=l(i)$. A formal definition is given below.
\[\begin{array}{rcl}
h(o(\mathbf{r})) & := & o(h(\mathbf{r})) \qquad\mbox{where $o\in Ops$} \\
h(\Exp[e]) & := & \Exp_{x\sim M[q]}[e] \qquad\mbox{where } M=\pair{\{M_{0}, M_{1}\}, Id} \mbox{ with }
	M_{0}=|0\>\<0|,\ M_{1}=|0\>\<1|, \\
	&  & \hfill \ x \mbox{ is fresh}  \\
        h(\Exp_{\bar{x}\sim M[\bar{q}]}[e]) & := &  \Exp_{\bar{x}\sim M'[\bar{q}]}[e] \quad  \mbox{where }  M'=\pair{\{M_{i0}, M_{i1}\}_i, \Lab'}\\
        & & \hfill \mbox{ with }
M_{i0}=M_i|0\>\<0|,\ M_{i1}=M_i|0\>\<1|,
  \Lab'(i0)=\Lab'(i1)=l(i)  \\
h(r_1 \comp r_2) & := & h(r_1) \comp h(r_2)\\
h(FO(P)) & := & FO(h(P))  \\
h(P_1\oplus P_2) & := & h(P_1) \oplus h(P_2)
\end{array}\]
 
 To ensure the freshness requirement on $x$ in $h(\Exp[e])$, we assume an enumeration of all the variables in $\Cvar$. Each time a fresh variable is needed, we take the next one which has not appeared in all the programs under consideration.
 
In rule [{\sf QUnit'}] we use the notation $g^U(P)$ for a syntactic substitution. It changes all 
$\Exp_{\bar{x}\sim M[\bar{q}]}[e]$ in $P$ into $\Exp_{\bar{x}\sim M'[\bar{q}]}[e]$, where $M=\pair{\{M_i \}_{i\in I},\Lab}$, $M'=\pair{\{M_i U\}_{i\in I},\Lab}$, and distributes over most other syntactic constructors of assertions. A formal definition is given below.
\[\begin{array}{rcl}
g^U(o(r)) & := & o(g^U(r)) \qquad\mbox{where $o\in Ops$} \\
g^U(\Exp[e]) & := & \Exp_{x\sim M[\bar{q}]}[e] \qquad \mbox{where } M=\pair{\{M_0\}, Id} \mbox{ with } M_0=U \mbox{ and $x$ is fresh}\\
g^U(\Exp_{\bar{x}\sim M[\bar{q}]}[e]) & := &  \Exp_{\bar{x}\sim M'[\bar{q}]}[e] \qquad  \mbox{where } M'=\pair{\{M_iU\}_{i\in I},\Lab} \\
g^U(r_1 \comp r_2) & := & g^U(r_1) \comp g^U(r_2)\\
g^U(FO(P)) & := & FO(g^U(P))  \\
g^U(P_1\oplus P_2) & := & g^U(P_1) \oplus g^U(P_2)
\end{array}\]

In rule [{\sf QMeas'}] we use the notation $f_{x,\bar{q}}^M(P)$ for a syntactic substitution. It changes all $\Exp[e]$ in $P$ into $\Exp_{x\sim M[\bar{q}]}[e]$. For the distribution expression  $\Exp_{\bar{y}\sim N[\bar{q}]}[e]$, it adds an outer layer of measurement to $N$. A formal definition is given below.
\[\begin{array}{rcl}
f_{x,\bar{q}}^M(o(r)) & := & o(f_{x,\bar{q}}^M(r)) \qquad\mbox{where $o\in Ops$} \\
f_{x,\bar{q}}^M(\Exp[e]) & := & \Exp_{x\sim M[\bar{q}]}[e]\\
f_{x,\bar{q}}^M(\Exp_{\bar{y}\sim N[\bar{q'}]}[e]) & := & \left\{\begin{array}{ll} 
\Exp_{x\bar{y}\sim M'[\bar{q}\cup\bar{q'}]}[e] \mbox{ with } M'=\pair{\{N_j M_i\}_{ij}, k'} \mbox{ and } k'(ij)=(k(i),l(j))  & \mbox{if } x\not\in\bar{y}\\
\Exp_{\bar{y}\sim M'[\bar{q}\cup\bar{q'}]}[e] \mbox{ with } M'=\pair{\{N_j M_i\}_{ij}, k'} \mbox{ and } k'(ij)=l(j)  & \mbox{if } x\in\bar{y}
\end{array}\right.\\
f_{x,\bar{q}}^M(r_1 \comp r_2) & := & f_{x,\bar{q}}^M(r_1) \comp f_{x,\bar{q}}^M(r_2)\\
f_{x,\bar{q}}^M(FO(P)) & := & FO(f_{x,\bar{q}}^M(P)) \\
f_{x,\bar{q}}^M(P_1\oplus P_2) & := & f_{x,\bar{q}}^M(P_1) \oplus f_{x,\bar{q}}^M(P_2)
\end{array}\]

We are going to show that the three functions $h(\cdot)$, $g^U(\cdot)$ and $f^M_{x,\bar{q}}(\cdot)$ behave well as they help to transform postconditions into preconditions for three kinds of commands: initialisation, applications of unitary operations, and measurements of quantum systems.
\begin{lemma}\label{lem:qinit}
The following two clauses hold.
\begin{enumerate}
        \item[(i)]  $\Denote{h(r)}_\povd = \Denote{r}_{\Denote{q:=|0\>}_\povd}$.
        \item[(ii)]  $\Denote{h(P)}_\povd \Rightarrow \Denote{P}_{\Denote{q:=|0\>}_\povd}$.
\end{enumerate}
\end{lemma}
\begin{proof}
We prove the  two statements by structural induction.
\begin{enumerate}
  \item[(i)] There are three cases for the structure of $r$.
   \begin{itemize}
   	 \item $r\equiv\Exp[e]$. We note that $\Denote{e}_\state = \Denote{e}_{\state[n/x]}$ for any number $n$ and fresh variable $x$ in the sense that $x$ does not appear in $e$.  Then we reason as follows.  
   	        \[\begin{array}{rcl} 
   	  \Denote{h(r)}_\povd & = & \Denote{\Exp_{x\sim M}[e]}_\povd \qquad\mbox{where } M=\pair{\{M_{0}, M_{1}\}, Id} \mbox{ with }
   	  M_{0}=|0\>\<0|,\ M_{1}=|0\>\<1|, \\
   	  &  & \hfill \ x \mbox{ is fresh}  \\
   	  & = & \sum_\state (|0\>\<0|\povd(\state)|0\>\<0| \cdot \Denote{e}_{\state[0/x]}  + |0\>\<1|\povd(\state) |1\>\<0|\cdot \Denote{e}_{\state[1/x]} ) \\
  	  & = & \sum_\state (|0\>\<0|\povd(\state)|0\>\<0| \cdot \Denote{e}_{\state}  + |0\>\<1|\povd(\state) |1\>\<0|\cdot \Denote{e}_{\state} ) \qquad x \mbox{ is fresh}\\
   	  & = &  \sum_\state (|0\>\<0|\povd(\state)|0\>\<0|   + |0\>\<1|\povd(\state) |1\>\<0|)\cdot \Denote{e}_{\state}  \\
     & = &  \sum_\state \Denote{q:=|0\>}_\povd(\state)\cdot \Denote{e}_{\state}  \\ 	  
   	  & = & \Denote{\Exp[e]}_{\Denote{q:=|0\>}_\povd} \\
   	  & = & \Denote{r}_{\Denote{q:=|0\>}_\povd}
   	  \end{array}\]
   	  
   	  \item $r\equiv \Exp_{\bar{x}\sim M[\bar{q}]}[e]$ for some $M=\pair{\{M_i\}_i, \Lab}$. Then
   	    	        \[\begin{array}{rcl} 
   	                \Denote{h(r)}_\povd & = &\Denote{\Exp_{\bar{x}\sim M'[\bar{q}]}[e]}_\povd \quad  \mbox{where }  M'=\pair{\{M_{i0}, M_{i1}\}_i, \Lab'} \\
                        & & \hspace{2.5cm} \mbox{ with }
   	  M_{i0}=M_i|0\>\<0|,\ M_{i1}=M_i|0\>\<1|,
   	   \Lab'(i0)=\Lab'(i1)=l(i)  \\
   	  & = & \sum_\state \sum_i (M_{i}|0\>\<0|\povd(\state)|0\>\<0|M^\dag_{i}    + M_{i} |0\>\<1|\povd(\state) |1\>\<0| M^\dag_{i})\cdot \Denote{e}_{\state[l(i)/\bar{x}]}  \\
   	  & = & \sum_\state \sum_i (M_{i}(|0\>\<0|\povd(\state)|0\>\<0|  + |0\>\<1|\povd(\state) |1\>\<0|) M^\dag_{i}\cdot \Denote{e}_{\state[l(i)/\bar{x}]}  \\   
     & = & \sum_\state \sum_i M_{i} \Denote{q:=|0\>}_\povd(\state) M^\dag_{i}\cdot \Denote{e}_{\state[l(i)/\bar{x}]}  \\   
   	& = & \Denote{\Exp_{\bar{x}\sim M[\bar{q}]}[e]}_{\Denote{q:=|0\>}_\povd} \\
   	& = & \Denote{r}_{\Denote{q:=|0\>}_\povd} 
   	  \end{array}\]
   	  
   	  \item $r\equiv o(r_1,...,r_k)$. The case can be proved by induction.
   	         \[\begin{array}{rcl} 
   	   \Denote{h(r)}_\povd & = & \Denote{o(h(r_1),...,h(r_k)}_\povd \\
   	   & = & o(\Denote{h(r_1)}_\povd,...,\Denote{h(r_k)}_\povd) \\
   	   & = & o(\Denote{r_1}_{\Denote{q:=|0\>}_\povd},..., \Denote{r_k}_{\Denote{q:=|0\>}_\povd}) \\
   	   & = & \Denote{o(r_1,...,r_k)}_{\Denote{q:=|0\>}_\povd} \\
   	   & = & \Denote{r}_{\Denote{q:=|0\>}_\povd} 
   	   \end{array}\]
   	\end{itemize}
 \item[(ii)] There are three cases for the structure of $P$.
    \begin{itemize}
     \item $P\equiv r_1\comp r_2$. In this case, we need to use statement (iii).
         \[\begin{array}{rcl} 
       \Denote{h(P)}_\povd & = & \Denote{h(r_1) \comp h(r_2)}_\povd \\
       & = & \Denote{h(r_1)}_\povd \comp \Denote{h(r_2)}_\povd \\
       & = & \Denote{r_1}_{\Denote{q:=|0\>}_\povd} \comp \Denote{r_2}_{\Denote{q:=|0\>}_\povd}) \\
       & = & \Denote{r_1 \comp r_2}_{\Denote{q:=|0\>}_\povd} \\
       & = & \Denote{r}_{\Denote{q:=|0\>}_\povd} 
       \end{array}\]
     \item $P\equiv P_1\oplus P_2$. This case is proved by induction.
           \[\begin{array}{rcl} 
       \Denote{h(P)}_\povd & = & \Denote{h(P_1) \oplus h(P_2)}_\povd \\
       & = & \exists \povd_1, \povd_2.\ \povd=\povd_1 + \povd_2 \wedge 
       \Denote{h(P_1)}_{\povd_1} \wedge \Denote{h(P_2)}_{\povd_2} \\
       & \Rightarrow & \exists \povd_1, \povd_2.\ \Denote{q:=|0\>}_\povd=\Denote{q:=|0\>}_{\povd_1} + \Denote{q:=|0\>}_{\povd_2} \\&  & \wedge 
       \Denote{P_1}_{\Denote{q:=|0\>}_{\povd_1}} \wedge \Denote{P_2}_{\Denote{q:=|0\>}_{\povd_2}}  \qquad\mbox{by Lemma~\ref{lem:split}}\\
       & = & \Denote{P_1\oplus P_2}_{\Denote{q:=|0\>}_\povd} \\
       & = & \Denote{r}_{\Denote{q:=|0\>}_\povd} 
       \end{array}\]
       \item $P\equiv FO(P_1,...,P_k)$. Again, this case is proved by induction.
           \[\begin{array}{rcl} 
       \Denote{h(P)}_\povd & = & \Denote{FO(h(P_1),..., h(P_k)}_\povd \\
       & = & FO(\Denote{h(P_1)}_\povd,..., \Denote{h(P_k)}_\povd) \\
       & = & FO(\Denote{P_1}_{\Denote{q:=|0\>}_\povd},..., \Denote{P_k}_{\Denote{q:=|0\>}_\povd} )\\
       & = & \Denote{FO(P_1,...,P_k)}_{\Denote{q:=|0\>}_\povd} \\
       & = & \Denote{P}_{\Denote{q:=|0\>}_\povd}
       \end{array}\]
    \end{itemize}    
\end{enumerate}
\end{proof}

\begin{lemma}\label{lem:qunit}
 The following two clauses hold.
	\begin{enumerate}
		\item[(i)]  $\Denote{g^U(r)}_\povd = \Denote{r}_{\Denote{U[\bar{q}]}_\povd}$.
		\item[(ii)] $\Denote{g^U(P)}_\povd \Rightarrow \Denote{P}_{\Denote{U[\bar{q}]}_\povd}$.
	\end{enumerate}
\end{lemma}
\begin{proof}
	The proof is similar to that of Lemma~\ref{lem:qinit} except for the treatment of two cases for statement (i).
	\begin{itemize}
	 \item $r\equiv \Exp[e]$. We infer that
	 \[\begin{array}{rcl} 
	 \Denote{g^U(r)}_\povd & = & \Denote{\Exp_{x\sim M[\bar{q}]}[e]}_\povd \qquad \mbox{where } M=\pair{\{M_0\}, Id} \mbox{ with } M_0=U \mbox{ and $x$ is fresh}\\
   & = & \sum_\state U\povd(\state)U^\dag \cdot\Denote{e}_{\state[0/x]}\\
   & = &  \sum_\state U\povd(\state)U^\dag \cdot\Denote{e}_{\state}\\
    & = &  \sum_\state \Denote{U[\bar{q}]}_\povd(\state) \cdot\Denote{e}_{\state}\\
    & = & \Denote{\Exp[e]}_{\Denote{U[\bar{q}]}_\povd}\\
    & = & \Denote{r}_{\Denote{U[\bar{q}]}_\povd}
	 \end{array}\]
	 
	 \item $r\equiv\Exp_{\bar{x}\sim M[\bar{q}]}[e]$. Suppose $M=\pair{\sset{M_i}_{i\in I},\Lab}$. 
	 We reason as follows.
	  \[\begin{array}{rcl} 
	 \Denote{g^U(r)}_\povd & = & \Denote{\Exp_{\bar{x}\sim M'[\bar{q}]}[e]}_\povd \quad\mbox{where } M'=\pair{\{M_i U\}_{i\in I},\Lab} \\
	 & = & \sum_\state\sum_{i}M_iU\povd(\state)U^\dag M^\dag_i\cdot\Denote{e}_{\state[l(i)/\bar{x}]} \\
	 & = & \sum_\state\sum_{i}M_i\Denote{U[\bar{q}]}_\povd(\state) M^\dag_i\cdot\Denote{e}_{\state[l(i)/\bar{x}]} \\
	 & = & \Denote{\Exp_{\bar{x}\sim M[\bar{q}]}[e]}_{\Denote{U[\bar{q}]}_\povd}\\
	 & = & \Denote{r}_{\Denote{U[\bar{q}]}_\povd}
	 \end{array}\]
	\end{itemize}
\end{proof}

\begin{lemma}\label{lem:assgn}
Let $a$ be an arithmetic expression, $\state$ be any state and $\state':=\state[\Denote{a}_\state/x]$. 
The following four clauses hold, where $x$ is not a bound variable  in $e$, $\psi$, $r$ and $P$.
	\begin{enumerate}
		\item[(i)] $\Denote{e[a/x]}_\state =\Denote{e}_{\state'}$
		\item[(ii)]  $\Denote{\psi[a/x]}_\state = \Denote{\psi}_{\state'}$
		\item[(iii)]  $\Denote{r[a/x]}_\povd = \Denote{r}_{\Denote{x:=a}_\povd}$.
		\item[(iv)] $\Denote{P[a/x]}_\povd \Rightarrow \Denote{P}_{\Denote{x:=a}_\povd}$.
	\end{enumerate}
\end{lemma}
\begin{proof}
The proof is similar to that of Lemma~\ref{lem:qinit}. As an example, we only consider one case for statement (iii).

Suppose $r\equiv\Exp_{\bar{y}\sim M[\bar{q}]}[e]$ with $M=\pair{\{M_i\}_i,l}$. There are two possibilities. 
\begin{itemize}
\item $x\in \bar{y}$. In this case, $x$ is a bound variable in $r$, which contradicts our assumption.
 \item $x\not\in \bar{y}$.    Notice that
 \begin{equation}\label{eq:a}
 \Denote{x:=a}_\povd(\state')\ =\ \sum_\state\{\povd(\state)\mid \state[\Denote{a}_\state/x]=\state'\}
 \end{equation}
 holds for any $\povd$ and $\state'$.
 We reason as follows.
 \[\begin{array}{rcl} 
 \Denote{r[a/x]}_\povd & = & \Denote{\Exp_{\bar{y}\sim M[\bar{q}]}[e[a/x]]}_\povd \\
 & = & \sum_\state\sum_iM_i\povd(\state)M^\dag_i\Denote{e[a/x]}_{\state[l(i)/\bar{y}]}\\
 & = & \sum_\state\sum_iM_i\povd(\state)M^\dag_i\Denote{e}_{\state[l(i)/\bar{y}][\Denote{a}_{\state[l(i)/\bar{y}]}/x]}\qquad\mbox{by statement (i)}\\
 & = & \sum_iM_i \sum_\state\povd(\state)M^\dag_i\Denote{e}_{\state[\Denote{a}_\state/x][l(i)/\bar{y}]}\\
 & = & \sum_iM_i \sum_{\state'} \sum_\state\{\povd(\state)\mid \state[\Denote{a}_\state/x]=\state'\}M^\dag_i\Denote{e}_{\state'[l(i)/\bar{y}]}\\
& = & \sum_iM_i \sum_{\state'} \Denote{x:=a}_\povd(\state')M^\dag_i\Denote{e}_{\state'[l(i)/\bar{y}]}\qquad\mbox{by }(\ref{eq:a})\\ 
& = &  \sum_{\state'} \sum_iM_i \Denote{x:=a}_\povd(\state')M^\dag_i\Denote{e}_{\state'[l(i)/\bar{y}]}\\ 
& = & \Denote{\Exp_{\bar{y}\sim M[\bar{q}]}[e]}_{\Denote{x:=a}_\povd}\\
 & = & \Denote{r}_{\Denote{x:=a}_\povd}
  \end{array}\]
\end{itemize}
\end{proof}

\begin{lemma}\label{lem:meas}
	\begin{enumerate}
		\item[(i)]  $\Denote{f_{x,\bar{q}}^M(r)}_\povd = \Denote{r}_{\Denote{x:=M[\bar{q}]}_\povd}$.
		\item[(ii)] $\Denote{f_{x,\bar{q}}^M(P)}_\povd \Rightarrow \Denote{P}_{\Denote{x:=M[\bar{q}]}_\povd}$.
	\end{enumerate}
\end{lemma}
\begin{proof}
We consider two cases for statement (i); the other cases are easier.  Assume that $M=\pair{\{M_i\}_i, k}$.
\begin{itemize}
	\item $r\equiv \Exp[e]$.
We reason as follows.
  \[\begin{array}{rcl}
\Denote{f_{x,\bar{q}}^M(r)}_\povd & = & \Denote{\Exp_{x\sim M[\bar{q}]}[e]}_\povd \\
& = & \sum_\state\sum_i M_i\povd(\state)M_i^\dag\cdot\Denote{e}_{\state[l(i)/x]} \\
& = & \sum_{\state'}\sum_\state\sum_i \{M_i\povd(\state)M_i^\dag\mid \state[l(i)/x]=\state'\}\cdot\Denote{e}_{\state'} \\
& = &  \sum_{\state'}\sum_\state\povd_\state(\state')\cdot\Denote{e}_{\state'} \qquad\mbox{where } \povd_\state(\state')=\sum_i \{M_i\povd(\state)M_i^\dag\mid \state[l(i)/x]=\state'\}\\
& = &\Denote{\Exp[e]}_{\sum_\state\povd_\state}\\
& = & \Denote{\Exp[e]}_{\Denote{x:=M[\bar{q}]}_\povd} \\
& = & \Denote{r}_{\Denote{x:=M[\bar{q}]}_\povd} 
\end{array}\]
The second last equality holds because $\Denote{x:=M[\bar{q}]}_{(\state,\povd(\state))}=\povd_\state$.

\item $r\equiv \Exp_{\bar{y}\sim N[\bar{q'}]}[e]$. There are two possibilities. Let us first assume that $x\not\in\bar{y}$ and $N=\pair{\{N_j\}_j, l}$.
  \[\begin{array}{rcl}
\Denote{f_{x,\bar{q}}^M(r)}_\povd & = & \Denote{\Exp_{x\bar{y}\sim M'[\bar{q}\cup\bar{q'}]}[e]}_\povd \mbox{ with } M'=\pair{\{N_j M_i\}_{ij}, k'} \mbox{ and } k'(ij)=(k(i),l(j))  \\
& = & \sum_\state\sum_{ij} N_jM_i\povd(\state)M_i^\dag N_j^\dag \cdot\Denote{e}_{\state[k(i),l(j)/x\bar{y}]} \\
& = & \sum_\state\sum_{j} N_j(\sum_i M_i\povd(\state)M_i^\dag) N_j^\dag \cdot\Denote{e}_{\state[k(i)/x][l(j)/\bar{y}]} \\
& = & \sum_{\state'}\sum_\state\sum_j N_j(\sum_i\{M_i\povd(\state)M_i^\dag \mid \state[k(i)/x]=\state'\})N_j^\dag\cdot\Denote{e}_{\state'[l(j)/\bar{y}]}\\
& = & \sum_{\state'}\sum_\state\sum_{j} N_j\povd_\state(\state') N_j^\dag \cdot\Denote{e}_{\state'[l(j)/\bar{y}]} \\
& & \qquad\mbox{where } \povd_\state(\state')=\sum_i \{M_i\povd(\state)M_i^\dag\mid \state[k(i)/x]=\state'\}\\
& = & \sum_{\state'}\sum_{j} N_j\sum_\state\povd_\state(\state') N_j^\dag \cdot\Denote{e}_{\state'[l(j)/\bar{y}]} \\
& = &\Denote{\Exp_{\bar{y}\sim N[\bar{q'}]}[e]}_{\sum_\state\povd_\state}\\
& = & \Denote{\Exp_{\bar{y}\sim N[\bar{q'}]}[e]}_{\Denote{x:=M[\bar{q}]}_\povd} \\
& = & \Denote{r}_{\Denote{x:=M[\bar{q}]}_\povd} 
\end{array}\]
If $x\in\bar{y}$, the proof is similar by noting that $\state[k(i)/x][l(j)/\bar{y}]=\state[l(j)/\bar{y}]$.
\end{itemize}
\end{proof}

The next theorem states that the concrete proof system is sound.
\begin{theorem}
	Every judgement provable in $\mathcal{S}_c$ is valid.
\end{theorem}
\begin{proof}
	We only need to prove that the four new rules [{\sf Assgn'}], [{\sf QInit'}], [{\sf QUnit'}] and [{\sf QMeas'}] are sound, which follows from Lemmas~\ref{lem:qinit} - \ref{lem:meas}; the soundness of all other rules are already shown in Theorem~\ref{thm:sounda}.
\end{proof}


We can define a precondition calculus to help with syntactic reasoning. Given an assertion $P$ as a postcondition and a loop-free command $c$, we construct an assertion  as a precondition for $c$, written as
$pc(c, P)$.
The computation rules for preconditions are given in Figure~\ref{fig:pc}. 
\begin{figure}
	\[\begin{array}{rcl}
pc(\Skip, P) & := & P \\
pc(x := a, P) & := & P[a/x] \\
pc(c_0; c_1, P) & := & pc(c_0, pc(c_1, P))\\
pc(\If b \Then c_0 \Else c_1, P) & := & 
(pc(c_0,P)\wedge\Box b)\oplus (pc(c_1,P)\wedge\Box \neg b)
 \\
pc(\abort, P) & := & \left\{\begin{array}{ll}
                                     \top & \mbox{if } P=\Box\false\\
                                     \mbox{undefined} & \mbox{otherwise}
                                     \end{array}\right. 
\\ 
pc(q:=|0\>, P) & := & h(P)\\
pc(U[\bar{q}], P) &:=& g^U(P)\\
pc(x:=M[\bar{q}], P) & := & f_{x,\bar{q}}^M(P)

	\end{array}\]
	\caption{Precondition calculus 
	}\label{fig:pc}
\end{figure}

\begin{theorem}
	Let $c$ be a non-looping command. The following rule is derivable.
	\[
	\prooftree
	\justifies
	\triple{pc(c, P)}{c}{P}
	\using
	[\sf PC]
    \endprooftree
	\]
\end{theorem}
\begin{proof}
	We proceed by induction on the structure of $c$.
	\begin{itemize}
		\item $c\equiv\Skip$. Then $pc(c,P)=P$ and we have 
		$\vdash \triple{P}{c}{P}$ by rule {\sf [Skip]}.
		\item $c\equiv\abort$. Then $pc(c, P)$ is only defined for $P=\Box\false$. In this case, we can infer that $\vdash \triple{\top}{\abort}{\Box\false}$ by rule {\sf [Abort]}.
		\item $c\equiv x:=a$. Then $pc(c,P) = P[a/x]$ and we have $\vdash\triple{P[a/x]}{c}{P}$ by rule {\sf [Assgn']}.
		\item $c\equiv c_0;c_1$. Then $pc(c,P) = pc(c_0, pc(c_1,P))$. By induction, we have $\vdash\triple{pc(c_1,P)}{c_1}{P}$ and $\vdash\triple{pc(c_0,pc(c_1,P))}{c_0}{pc(c_1,P)}$. By using rule {\sf [Seq]}, we obtain $\vdash \triple{pc(c_0,pc(c_1,P))}{c}{P}$.
		\item $c\equiv \If b \Then c_0 \Else c_1$. Then $pc(c,P) = (pc(c_0,P)\wedge\Box b)\oplus (pc(c_1,P)\wedge\Box \neg b)$.  By induction, we have that $\vdash \triple{pc(c_0,P)}{c_0}{P}$ and $\vdash \triple{pc(c_1,P)}{c_1}{P}$. It is obvious that $pc(c_0,P)\wedge \Box b\Rightarrow pc(c_0,P)$. We can use rule {\sf [Conseq]} to infer that $\vdash\triple{pc(c_0,P)\wedge \Box b}{c_0}{P}$. Similarly, we have $\vdash\triple{pc(c_1,P)\wedge \Box \neg b}{c_1}{P}$.
		By applying rule {\sf [Cond]}, we have that
		$\vdash\triple{pc(c, P\oplus P)}{c}{P\oplus P}$.
		Since $P\oplus P\Leftrightarrow P$, we use rule {\sf [Conseq]} again to infer the required result that $\vdash\triple{pc(c, P)}{c}{P}$.
		\item $c\equiv q:=|0\>$.  A direct consequence of rule {\sf [QInit']}.
		\item $c\equiv U[\bar{q}]$. By using rule {\sf [QUnit']}.
		\item $c\equiv M[\bar{q}]$. By using rule {\sf [QMeas']}.
	\end{itemize}
\end{proof}

\section{Example: superdense coding}\label{sec:example}
In this section, we illustrate the use of the proof system $\mathcal{S}_c$ via the example of superdense coding.

Superdense coding was proposed by Bennett and Wiesner in 1992~\cite{BW92}. It is a quantum communication protocol allowing two classical bits to be encoded in one qubit during a transmission, so it needs only one quantum channel. Such advantage is based on the use of a maximally
entangled state, EPR state. An EPR state can be transformed into all the four kinds of EPR states through 1-qubit operations, and these EPR states are mutually orthogonal. 
\paragraph*{Protocol.}
We suppose the sender and the receiver of the communication are Alice and Bob, then the protocol goes as follows:
\begin{enumerate}
	\item Alice and Bob prepare an EPR state $\frac{|00\>+|11\>}{\sqrt{2}}$ together. Then they share the qubits, Alice holding $q_0$ and Bob holding $q_1$.
	\item Depending on the message Alice wants to send, she applies a gate to her qubit. If Alice wants to send $00$, she does nothing. If Alice wants to send $01$, she applies the X gate. To send $10$, she applies the Z gate. To send $11$, she applies both $X$ and $Z$.
	\item Then Alice sends the qubit $q_0$ to Bob.
	\item Bob applies a CNOT operation on $q_0,q_1$ and a Hadamard operation on $q_0$ to remove the entanglement.
	\item Bob measures $q_0$ and $q_1$ to get the message.
\end{enumerate}
After the execution of the protocol above, Bob gets the value that  Alice wants to send. The protocol exactly transmits two classical bits of information by sending one qubit from Alice to Bob. A quantum circuit implementing the protocol is illustrated in Figure~\ref{fig:sc}.
	\begin{figure}[t]
		\[
		\Qcircuit @C=1.2em @R=2em {
			\lstick{x_0} & \cw & \cw & \cw & \cw & \cw & \cw & \cw & \cctrl{2} & & & & & & & & & \\
			\lstick{x_1} & \cw & \cw & \cw & \cw & \cw & \cctrl{1} &  & & & & & & & & & & \\
			\lstick{q_0} & \qw & \gate{H} & \qw & \ctrl{1} & \qw & \gate{X} & \qw & \gate{Z} & \qw & \ctrl{1} & \qw & \gate{H} & \qw & \meter & \cw & \cw & y_0 \\
			\lstick{q_1} & \qw & \qw & \qw & \targ & \qw & \qw & \qw & \qw & \qw & \targ & \qw & \qw & \qw & \meter & \cw & \cw & y_1 \\
		}
		\]
\caption{Superdense coding}\label{fig:sc}
	\end{figure}
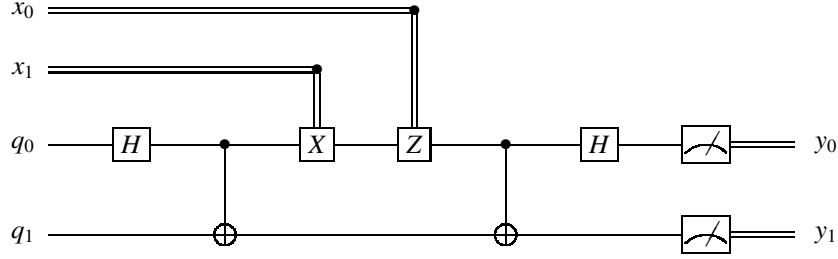

The protocol can also be described by the quantum program $SC$ given in Figure~\ref{fig:prog}, where for any pure state $|\varphi\>$, we write $[|\varphi\>]$ for its density operator $|\varphi\>\<\varphi|$.
\begin{figure}
	\[\begin{array}{cl}
	SC \equiv & \\
1: &	q_0 := |0\>; \\
2: &	q_1 := |0\>; \\
3: &	H[q_0]; \\
4: &	CNOT[q_0 q_1]; \\
5: &	\If x_1=1 \Then X[q_0]; \\
6: &	\If x_0=1 \Then Z[q_0]; \\
7: &	CNOT[q_0 q_1]; \\
8: &	H[q_0]; \\
9: &	y_0 :=M[q_0]; \\
10: &	y_1 :=M[q_1]  \\
&	\quad\mbox{where } M=\pair{\{M_0, M_1\}, Id}, M_0=[|0\>], M_1=[|1\>]
	\end{array}\]
	\caption{The quantum program of implementing superdense coding}\label{fig:prog}
\end{figure}

According to the operational rules in Figure~\ref{fig:exec}, we can derive the following sequence of transitions, where  the initial values of the four classical variables in the first configuration can be arbitrary and we use $*$ to stand for unimportant commands or the values of variables.
\[\begin{array}{ll}
& (SC, x_0x_1y_0y_1, [|00\>])\VS\\
\rightarrow & (*, *, [\frac{|0\>+|1\>}{\sqrt{2}}|0\>]) \VS\\
\rightarrow & (*, *, [\frac{|00\>+|11\>}{\sqrt{2}}]) \VS\\
\rightarrow & (*, *, [X_0^{x_1}\frac{|00\>+|11\>}{\sqrt{2}}]) \VS\\
\rightarrow & (*, *, [Z_0^{x_0}X_0^{x_1}\frac{|00\>+|11\>}{\sqrt{2}}]) \VS\\
\equiv & \left\{\begin{array}{ll}
(*, 00y_0y_1, [\frac{|00\>+|11\>}{\sqrt{2}}]) & \mbox{ if } x_0=x_1=0 \\
(*, 01y_0y_1, [\frac{|10\>+|01\>}{\sqrt{2}}]) & \mbox{ if } x_0=0, x_1=1 \\
(*, 10y_0y_1, [\frac{|00\>-|11\>}{\sqrt{2}}]) & \mbox{ if } x_0=1, x_1=0 \\
(*, 11y_0y_1, [\frac{|10\>-|01\>}{\sqrt{2}}]) & \mbox{ if } x_0=x_1=1 
             \end{array} \right.\VS\\
\rightarrow & \left\{\begin{array}{l}
             (*, 00y_0y_1, [\frac{|00\>+|10\>}{\sqrt{2}}])  \\
             (*, 01y_0y_1, [\frac{|11\>+|01\>}{\sqrt{2}}])  \\
             (*, 10y_0y_1, [\frac{|00\>-|10\>}{\sqrt{2}}])  \\
             (*, 11y_0y_1, [\frac{|11\>-|01\>}{\sqrt{2}}])  
             \end{array} \right.\VS\\
\rightarrow & \left\{\begin{array}{l}
             (*, 00y_0y_1, [|00\>])  \\
             (*, 01y_0y_1, [|01\>])  \\
             (*, 10y_0y_1, [|10\>])  \\
             (*, 11y_0y_1, [|11\>])  
             \end{array} \right.\VS\\
\rightarrow & \left\{\begin{array}{l}
             (*, 000y_1, [|00\>])  \\
             (*, 010y_1, [|01\>])  \\
             (*, 101y_1, [|10\>])  \\
             (*, 111y_1, [|11\>])  
             \end{array} \right.\VS\\
\rightarrow & \left\{\begin{array}{l}
             (\nil, 0000, [|00\>])  \\
             (\nil, 0101, [|01\>])  \\
             (\nil, 1010, [|10\>])  \\
             (\nil, 1111, [|11\>])  
             \end{array} \right.\\
\end{array}\]

We observe that in each case of the four last configurations, we always have the value of $x_0x_1$ coincide with $y_0y_1$ as expected. Indeed, we would like to show that the judgement
\begin{equation}\label{eq:correct}
\triple{\true} {SC}{\Box(x_0=y_0 \wedge x_1=y_1)}
\end{equation}
is provable in our concrete proof system. This can be accomplished by a sequence of derivations; for every line of command in Figure~\ref{fig:prog} we need to prove a Hoare triple. We  start from line 10 and proceed backwards. The first six steps can be derived by using the rules {\sf[QMeas']}, {\sf [QUnit']}, {\sf [Cond]} and {\sf [Split]}, as shown in Figure~\ref{fig:cond}. Continue the reasoning until line 1, we obtain the following precondition for SC.
\[\begin{array}{l}
  \{(\Exp_{y_0y_1\sim M_{10}[q_0 q_1]}[\Char_{\psi\wedge x_0=1\wedge x_1=1}]= \Exp_{y_0y_1\sim M_{10}[q_0 q_1]}[\Char_{\true}]) \\
 \oplus (\Exp_{y_0y_1\sim M_9[q_0 q_1]}[\Char_{\psi\wedge x_0=1\wedge x_1=0}]= \Exp_{y_0y_1\sim M_9[q_0 q_1]}[\Char_{\true}]) \qquad\qquad (\dag)\\
\oplus (\Exp_{y_0y_1\sim M_8[q_0 q_1]}[\Char_{\psi\wedge x_0=0\wedge x_1=1}]= \Exp_{y_0y_1\sim M_8[q_0 q_1]}[\Char_{\true}]) \\
\oplus (\Exp_{y_0y_1\sim M_7[q_0 q_1]}[\Char_{\psi\wedge x_0=0\wedge x_1=0}]= \Exp_{y_0y_1\sim M_7[q_0 q_1]}[\Char_{\true}]) \},\\
\mbox{where } M_{10}\equiv \pair{\{E_{0000}, E_{0001},..., E_{1111}\}, f}\mbox{ with}\\
 \hspace{2.3cm}  	E_{0000}\equiv E_{00}H_{q_0}CNOT_{q_0q_1}Z_{q_0}X_{q_0}CNOT_{q_0q_1}H_{q_0}|0\>_{q_1}\<0|\cdot |0\>_{q_0}\<0|, \\ 
\hspace{2.3cm} 	 E_{0001}\equiv E_{00}H_{q_0}CNOT_{q_0q_1}Z_{q_0}X_{q_0}CNOT_{q_0q_1}H_{q_0}|0\>_{q_1}\<0|\cdot |0\>_{q_0}\<1|, \\ 
 \hspace{2.3cm} E_{0010}\equiv E_{00}H_{q_0}CNOT_{q_0q_1}Z_{q_0}X_{q_0}CNOT_{q_0q_1}H_{q_0}|0\>_{q_1}\<1|\cdot |0\>_{q_0}\<0|,\\  \hspace{2.3cm} E_{0011}\equiv E_{00}H_{q_0}CNOT_{q_0q_1}Z_{q_0}X_{q_0}CNOT_{q_0q_1}H_{q_0}|0\>_{q_1}\<1|\cdot |0\>_{q_0}\<1|,  \\
  \hspace{2.3cm}  	E_{0100}\equiv E_{01}H_{q_0}CNOT_{q_0q_1}Z_{q_0}X_{q_0}CNOT_{q_0q_1}H_{q_0}|0\>_{q_1}\<0|\cdot |0\>_{q_0}\<0|, \\ 
    \hspace{3.3cm}   \vdots \\
 \hspace{2.3cm} E_{1111}\equiv E_{11}H_{q_0}CNOT_{q_0q_1}Z_{q_0}X_{q_0}CNOT_{q_0q_1}H_{q_0}|0\>_{q_1}\<1|\cdot |0\>_{q_0}\<1|,  \\
  \hspace{2.3cm} f(00**)=00,\ f(01**)=01,\ f(10**)=10,\ f(11**)=11\\  	
 M_{9}\equiv \pair{\{E'_{0000}, E'_{0001},..., E'_{1111}\}, f}\mbox{ with}\\ 	
\hspace{1.3cm}	E'_{0000}\equiv E_{00}H_{q_0}CNOT_{q_0q_1}Z_{q_0}CNOT_{q_0q_1}H_{q_0}|0\>_{q_1}\<0|\cdot |0\>_{q_0}\<0|, \\ 
    \hspace{2.3cm}  \vdots \\
\hspace{1.3cm} E'_{1111}\equiv E_{11}H_{q_0}CNOT_{q_0q_1}Z_{q_0}CNOT_{q_0q_1}H_{q_0}|0\>_{q_1}\<1|\cdot |0\>_{q_0}\<1|  \\
 M_{8}\equiv \pair{\{E''_{0000}, E''_{0001},..., E''_{1111}\}, f}\mbox{ with}\\ 	
\hspace{1.3cm}	E''_{0000}\equiv E_{00}H_{q_0}CNOT_{q_0q_1}X_{q_0}CNOT_{q_0q_1}H_{q_0}|0\>_{q_1}\<0|\cdot |0\>_{q_0}\<0|, \\ 
\hspace{2.3cm}   \vdots \\
\hspace{1.3cm} E''_{1111}\equiv E_{11}H_{q_0}CNOT_{q_0q_1}X_{q_0}CNOT_{q_0q_1}H_{q_0}|0\>_{q_1}\<1|\cdot |0\>_{q_0}\<1| \\

 M_{7}\equiv \pair{\{E'''_{0000}, E'''_{0001},..., E'''_{1111}\}, f}\mbox{ with}\\ 	
\hspace{1.3cm}	E'''_{0000}\equiv E_{00}H_{q_0}CNOT_{q_0q_1}CNOT_{q_0q_1}H_{q_0}|0\>_{q_1}\<0|\cdot |0\>_{q_0}\<0|, \\ 
\hspace{2.3cm}  \vdots\\
\hspace{1.3cm} E'''_{1111}\equiv E_{11}H_{q_0}CNOT_{q_0q_1}CNOT_{q_0q_1}H_{q_0}|0\>_{q_1}\<1|\cdot |0\>_{q_0}\<1|  \\
\end{array}\]

The measurement operators in $M_{10}$ look complicated. However, a simple calculation shows that among the 16 operators only four of them are non-zero. Indeed, the simplified form of $M_{10}$ is 
\[M_{10}=\pair{\{E_{1100}, E_{1101}, E_{1110}, E_{1111}\}, f}\]
where 
\[E_{1100}=\begin{bmatrix}
0 & 0 & 0 & 0\\
0 & 0 & 0 & 0\\
0 & 0 & 0 & 0\\
1 & 0 & 0 & 0
\end{bmatrix}
\quad E_{1101}=\begin{bmatrix}
0 & 0 & 0 & 0\\
0 & 0 & 0 & 0\\
0 & 0 & 0 & 0\\
0 & 0 & 1 & 0
\end{bmatrix}
\quad E_{1110}=\begin{bmatrix}
0 & 0 & 0 & 0\\
0 & 0 & 0 & 0\\
0 & 0 & 0 & 0\\
0 & 1 & 0 & 0
\end{bmatrix}
\quad E_{1111}=\begin{bmatrix}
0 & 0 & 0 & 0\\
0 & 0 & 0 & 0\\
0 & 0 & 0 & 0\\
0 & 0 & 0 & 1
\end{bmatrix}
\]
Similarly, the simplified form of $M_{9}$ is 
$M_{9}=\pair{\{E'_{1000}, E'_{1001}, E'_{1010}, E'_{1011}\}, f}$,
where 
\[E'_{1000}=\begin{bmatrix}
0 & 0 & 0 & 0\\
0 & 0 & 0 & 0\\
1 & 0 & 0 & 0\\
0 & 0 & 0 & 0
\end{bmatrix}
\quad E'_{1001}=\begin{bmatrix}
0 & 0 & 0 & 0\\
0 & 0 & 0 & 0\\
0 & 0 & 1 & 0\\
0 & 0 & 0 & 0
\end{bmatrix}
\quad E'_{1010}=\begin{bmatrix}
0 & 0 & 0 & 0\\
0 & 0 & 0 & 0\\
0 & 1 & 0 & 0\\
0 & 0 & 0 & 0
\end{bmatrix}
\quad E'_{1011}=\begin{bmatrix}
0 & 0 & 0 & 0\\
0 & 0 & 0 & 0\\
0 & 0 & 0 & 1\\
0 & 0 & 0 & 0
\end{bmatrix}
\]
The simplified form of $M_{8}$ is 
$M_{8}=\pair{\{E''_{0100}, E''_{0101}, E''_{0110}, E''_{0111}\}, f}$,
where 
\[E''_{0100}=\begin{bmatrix}
0 & 0 & 0 & 0\\
1 & 0 & 0 & 0\\
0 & 0 & 0 & 0\\
0 & 0 & 0 & 0
\end{bmatrix}
\quad E''_{0101}=\begin{bmatrix}
0 & 0 & 0 & 0\\
0 & 0 & 1 & 0\\
0 & 0 & 0 & 0\\
0 & 0 & 0 & 0
\end{bmatrix}
\quad E''_{0110}=\begin{bmatrix}
0 & 0 & 0 & 0\\
0 & 1 & 0 & 0\\
0 & 0 & 0 & 0\\
0 & 0 & 0 & 0
\end{bmatrix}
\quad E''_{0111}=\begin{bmatrix}
0 & 0 & 0 & 0\\
0 & 0 & 0 & 1\\
0 & 0 & 0 & 0\\
0 & 0 & 0 & 0
\end{bmatrix}
\]
Finally, the simplified form of $M_{7}$ is 
$M_{7}=\pair{\{E'''_{0000}, E'''_{0001}, E'''_{0010}, E'''_{0011}\}, f}$,
where 
\[E'''_{0000}=\begin{bmatrix}
1 & 0 & 0 & 0\\
0 & 0 & 0 & 0\\
0 & 0 & 0 & 0\\
0 & 0 & 0 & 0
\end{bmatrix}
\quad E'''_{0001}=\begin{bmatrix}
0 & 0 & 1 & 0\\
0 & 0 & 0 & 0\\
0 & 0 & 0 & 0\\
0 & 0 & 0 & 0
\end{bmatrix}
\quad E'''_{0010}=\begin{bmatrix}
0 & 1 & 0 & 0\\
0 & 0 & 0 & 0\\
0 & 0 & 0 & 0\\
0 & 0 & 0 & 0
\end{bmatrix}
\quad E'''_{0011}=\begin{bmatrix}
0 & 0 & 0 & 1\\
0 & 0 & 0 & 0\\
0 & 0 & 0 & 0\\
0 & 0 & 0 & 0
\end{bmatrix}
\]

\begin{figure}[t]
	\[\begin{array}{cl}
	& \cond{\{\true\}} \\
	SC \equiv & \\
	1: &	q_0 := |0\>; \\
	2: &	q_1 := |0\>; \\
	3: &	H[q_0]; \\
	4: &	CNOT[q_0 q_1]; \\
	 &  \cond{\{(\Exp_{y_0y_1\sim M_6[q_0 q_1]}[\Char_{\psi\wedge x_0=1\wedge x_1=1}]= \Exp_{y_0y_1\sim M_6[q_0 q_1]}[\Char_{\true}])} \\
	 	&	\cond{ \oplus (\Exp_{y_0y_1\sim M_4[q_0 q_1]}[\Char_{\psi\wedge x_0=1\wedge x_1=0}]= \Exp_{y_0y_1\sim M_4[q_0 q_1]}[\Char_{\true}]) }\\
	&	\cond{ \oplus (\Exp_{y_0y_1\sim M_5[q_0 q_1]}[\Char_{\psi\wedge x_0=0\wedge x_1=1}]= \Exp_{y_0y_1\sim M_5[q_0 q_1]}[\Char_{\true}]) }\\
		&	\cond{ \oplus (\Exp_{y_0y_1\sim M_3[q_0 q_1]}[\Char_{\psi\wedge x_0=0\wedge x_1=0}]= \Exp_{y_0y_1\sim M_3[q_0 q_1]}[\Char_{\true}]) \},}\\
	& \cond{\mbox{where } M_6\equiv\langle\{E_{00}H_{q_0}CNOT_{q_0q_1}Z_{q_0}X_{q_0}, \ E_{01}H_{q_0}CNOT_{q_0q_1}Z_{q_0}X_{q_0},} \\ 
		& \hspace{2.3cm} \cond{E_{10}H_{q_0}CNOT_{q_0q_1}Z_{q_0}X_{q_0},\ E_{11}H_{q_0}CNOT_{q_0q_1}Z_{q_0}X_{q_0}\}, Id\rangle ,} \\
	& \cond{M_5\equiv\langle\{E_{00}H_{q_0}CNOT_{q_0q_1}X_{q_0}, \ E_{01}H_{q_0}CNOT_{q_0q_1}X_{q_0},} \\ 
	& \hspace{2.3cm} \cond{E_{10}H_{q_0}CNOT_{q_0q_1}X_{q_0},\ E_{11}H_{q_0}CNOT_{q_0q_1}X_{q_0}\}, Id\rangle} \\
	5: &	\If x_1=1 \Then X[q_0]; \\
			 &  \cond{\{(\Exp_{y_0y_1\sim M_4[q_0 q_1]}[\Char_{\psi\wedge x_0=1}]= \Exp_{y_0y_1\sim M_4[q_0 q_1]}[\Char_{\true}])} \\
		&	\cond{ \oplus (\Exp_{y_0y_1\sim M_3[q_0 q_1]}[\Char_{\psi\wedge x_0=0}]= \Exp_{y_0y_1\sim M_3[q_0 q_1]}[\Char_{\true}]) \},}\\
	& \cond{\mbox{where } M_4\equiv\langle\{E_{00}H_{q_0}CNOT_{q_0q_1}Z_{q_0}, \ E_{01}H_{q_0}CNOT_{q_0q_1}Z_{q_0},} \\ 
	& \hspace{2.3cm} \cond{E_{10}H_{q_0}CNOT_{q_0q_1}Z_{q_0},\ E_{11}H_{q_0}CNOT_{q_0q_1}Z_{q_0}\}, Id\rangle} \\
	6: &	\If x_0=1 \Then Z[q_0]; \\
		 &  \cond{\{\Exp_{y_0y_1\sim M_3[q_0 q_1]}[\Char_{\psi}]= \Exp_{y_0y_1\sim M_3[q_0 q_1]}[\Char_{\true}] \}, }\\
	& \cond{\mbox{where } M_3\equiv\pair{\{E_{00}H_{q_0}CNOT_{q_0q_1}, E_{01}H_{q_0}CNOT_{q_0q_1}, E_{10}H_{q_0}CNOT_{q_0q_1}, E_{11}H_{q_0}CNOT_{q_0q_1}\}, Id}} \\
	7: &	CNOT[q_0 q_1]; \\
	 &  \cond{\{\Exp_{y_0y_1\sim M_2[q_0 q_1]}[\Char_{\psi}]= \Exp_{y_0y_1\sim M_2[q_0 q_1]}[\Char_{\true}] \},}\\
	& \cond{\mbox{where } M_2\equiv\pair{\{E_{00}H_{q_0}, E_{01}H_{q_0}, E_{10}H_{q_0}, E_{11}H_{q_0}\}, Id} } \\
	8: &	H[q_0]; \\
	 &  \cond{\{\Exp_{y_0y_1\sim M_1[q_0 q_1]}[\Char_{\psi}]= \Exp_{y_0y_1\sim M_1[q_0 q_1]}[\Char_{\true}] \}}\\
	 & \cond{\mbox{where } M_1\equiv\pair{\{E_{00}, E_{01}, E_{10}, E_{11}\}, Id}, } \\
	 & \cond{ E_{00}\equiv [|0\>]_{q_1}\cdot [|0\>]_{q_0}, \
	 	E_{01}\equiv [|1\>]_{q_1}\cdot [|0\>]_{q_0}, \
	 	E_{10}\equiv [|0\>]_{q_1}\cdot [|1\>]_{q_0}, \
	 	E_{11}\equiv [|1\>]_{q_1}\cdot [|1\>]_{q_0}
	  } \\
	9: &	y_0 :=M[q_0]; \\
	    &  \cond{\{\Exp_{y_1\sim M[q_1]}[\Char_{\psi}]= \Exp_{y_1\sim M[q_1]}[\Char_{\true}] \}}\\
	10: &	y_1 :=M[q_1]  \\
	&	\quad\mbox{where } M=\pair{\{E_0, E_1\}, Id}, E_0=[|0\>], E_1=[|1\>]\\
   & \cond{\{\Box\psi\}, \mbox{ where }\psi\equiv x_0=y_0 \wedge x_1=y_1 }
	\end{array}\]
	\caption{The quantum program with pre- and postconditions}\label{fig:cond}
\end{figure}

Write $P$ for the assertion in ($\dag$). We have seen that 
\begin{equation}\label{eq:b}
\triple{P}{SC}{\Box\psi} .
\end{equation}
We observe that $\true \Leftrightarrow P$. This can be seen as follows. Let 
\[\begin{array}{rcl}
P_{11} & \equiv & (\Exp_{y_0y_1\sim M_{10}[q_0 q_1]}[\Char_{\psi\wedge x_0=1\wedge x_1=1}]= \Exp_{y_0y_1\sim M_{10}[q_0 q_1]}[\Char_{\true}]) \\
P_{10} & \equiv & (\Exp_{y_0y_1\sim M_9[q_0 q_1]}[\Char_{\psi\wedge x_0=1\wedge x_1=0}]= \Exp_{y_0y_1\sim M_9[q_0 q_1]}[\Char_{\true}]) \\
P_{01} & \equiv & (\Exp_{y_0y_1\sim M_8[q_0 q_1]}[\Char_{\psi\wedge x_0=0\wedge x_1=1}]= \Exp_{y_0y_1\sim M_8[q_0 q_1]}[\Char_{\true}]) \\
P_{00} & \equiv & (\Exp_{y_0y_1\sim M_7[q_0 q_1]}[\Char_{\psi\wedge x_0=0\wedge x_1=0}]= \Exp_{y_0y_1\sim M_7[q_0 q_1]}[\Char_{\true}])\\
b_{11} & \equiv & x_0=1\wedge x_1=1\\
b_{10} & \equiv & x_0=1\wedge x_1=0\\
b_{01} & \equiv & x_0=0\wedge x_1=1\\
b_{00} & \equiv & x_0=0\wedge x_1=0
\end{array}\]
We have $P = P_{11}\oplus P_{10}\oplus P_{10}\oplus P_{00}$.  For any POVD $\povd$, it is easy to see that 
\[\povd = \povd_{|b_{11}} + \povd_{|b_{10}} + \povd_{|b_{01}} + \povd_{|b_{00}} \]
We have that $\povd_{|b_{11}} \models P_{11}$ because
\[\begin{array}{ll}
& \Denote{\Exp_{y_0y_1\sim M_{10}[q_0 q_1]}[\Char_{\psi\wedge x_0=1\wedge x_1=1}]}_{\povd_{|b_{11}} }  \\
= & \sum_{\state}\sum_i E_i\povd_{|b_{11}}(\state)M^\dag_i\cdot \Char_{\Denote{\psi\wedge x_0=1\wedge x_1=1}_{\state[f(i)/y_0y_1]}}\\
& \mbox{where } i\in\{1100, 1101, 1110, 1100\} \\
= & \sum_{\state}\sum_i E_i\povd_{|b_{11}}(\state)M^\dag_i\cdot \Char_{\Denote{\psi\wedge x_0=1\wedge x_1=1}_{\state[11/y_0y_1]}}\\
= & \sum_{\state}\sum_i E_i\povd_{|b_{11}}(\state)M^\dag_i\cdot 1 \\
= & \sum_{\state}\sum_i E_i\povd_{|b_{11}}(\state)M^\dag_i\cdot \Char_{\Denote{\true}_{\state[11/y_0y_1]}}\\
= & \Denote{\Exp_{y_0y_1\sim M_{10}[q_0 q_1]}[\Char_{\true}]}_{\povd_{|b_{11}} }  
\end{array}\]
which implies $\Denote{P_{11}}_{\povd_{|b_{11}} } =\true$.
Similarly, we can check that $\povd_{|b_{10}} \models P_{10}$, etc. Therefore, we obtain that $\mu \models P$. As $\povd$ is arbitrarily chosen, we have verified that $\true\Leftrightarrow P$. By (\ref{eq:b}) and rule {\sf [Conseq]}, we finally see that the triple $\triple{\true}{SC}{\Box\psi}$ is provable.

As we can see in (\ref{eq:correct}), the postcondition of that Hoare triple is an assertion about classical variables, even though quantum computation takes place during the execution of the program $SC$. For such scenario, it is very natural to prove the correctness of programs via an satisfaction-based proof system, rather than an expectation-based one.

\section{Conclusion and future work}\label{sec:clu}
We have introduced a simple quantum imperative  language that has both classical and quantum constructs by extending the language \textbf{IMP} studied in depth by Winskel. We have investigated its formal semantics by providing a small-step operational semantics, a denotational semantics and two Hoare-style proof systems: an abstract one and a concrete one. In order to define the semantics, we have used the notion of POVDs to represent the states of programs. Therefore, a program can be considered as a transformer of POVDs. Following the work of Barthe et al, we have designed 
two satisfaction-based proof systems, as opposed to the usual expectation-based systems. The abstract proof system turns out to be sound and relatively complete, while the concrete one is sound.

As to the future work, at least two immediate improvements are interesting and worth being considered.
\begin{itemize}
\item The proof rule {\sf [While]} is not satisfactory because it involves two sequences of assertions $(P_n)$ and $(P'_n)$. In either the purely classical \cite{hoare1969axiomatic} or purely quantum setting \cite{Yin12}, the rule can be elegantly formulated. However, in the presence of both classical and quantum variables, it remains a challenge to find a more concise formulation of the rule.
\item The reasoning in Section~\ref{sec:example} about the example of superdense coding was done manually. In the future, we would like to embed our program logic in a proof assistant so as to facilitate the reasoning.
  
\end{itemize}

\bibliographystyle{abbrv}
\bibliography{ref}

\end{document}